\documentclass[11pt,longbibliography]{article}

\usepackage[margin=0.75in]{geometry}
\usepackage{amsmath,amssymb,amsthm,mathtools}
\usepackage{hyperref}
\usepackage{booktabs} 
\usepackage{multirow} 
\usepackage{authblk}  
\usepackage[nottoc]{tocbibind}  
\usepackage{appendix} 
\hypersetup{colorlinks=true,allcolors=blue}
\usepackage{graphicx}

\newtheorem{definition}{Definition}[section]

\newtheorem{lemma}{Lemma}[section]
\newtheorem{theorem}{Theorem}[section]
\newtheorem{corollary}{Corollary}[section]
\newtheorem{remark}{Remark}[section]

\DeclareMathOperator{\Tr}{Tr}
\DeclareMathOperator{\E}{\mathbb{E}}
\DeclareMathOperator{\argmin}{arg\,min}
\DeclareMathOperator{\rankTT}{rank_{TT}}

\newcommand{\R}{\mathbb{R}}
\newcommand{\N}{\mathbb{N}}
\newcommand{\Nzero}{\mathbb{N}_0}

\newcommand{\cH}{\mathcal{H}}
\newcommand{\cA}{\mathcal{A}}
\newcommand{\cB}{\mathcal{B}}

\newcommand{\ip}[2]{\left\langle #1,#2\right\rangle}
\newcommand{\deff}{d_{\mathrm{eff}}}
\newcommand{\cG}{\mathcal{G}}

\begin{document}

\title{Local tensor-train surrogates for quantum learning models}
\author[1]{Sreeraj Rajindran Nair\thanks{Email: \texttt{sreeraj.r.nair@student.uts.edu.au}}}
\author[1]{Christopher Ferrie\thanks{Email: \texttt{christopher.ferrie@uts.edu.au}}}

\affil[1]{Centre for Quantum Software and Information, University of Technology Sydney, NSW 2007, Australia}

\date{}
\maketitle

\begin{abstract}

A key bottleneck in quantum machine learning is the computational cost of repeated quantum circuit evaluations during the inference phase. To address this, we present a framework for constructing fast, cheap, provably accurate classical tensor-train surrogates of fully trained quantum machine learning models within local patches of their input data space. The approach combines Taylor polynomial approximation with a tensor-train (TT) representation and embeds it in a statistical learning paradigm via empirical risk minimization. In our analysis, the Taylor-TT construction serves as a deterministic error certificate proving that the TT hypothesis class contains a good approximation; empirical risk minimization then provably recovers a surrogate with controlled generalization error and explicit bounds. This translates into three independently controllable error sources: (i) Taylor truncation error controlled by the patch radius $r$ and polynomial degree $p$, (ii) TT approximation error controlled by the bond dimension $\chi$, and (iii) statistical estimation error. While the parameter count scales polynomially in the number of data dimensions $N$, i.e., $\deff = N(p+1)\chi^2$ rather than the naive $(p+1)^N$, the worst-case constants inherit an exponential factor through the tensor-product feature norm during Taylor polynomial embedding onto TT. This cleanly separates representation complexity from feature-induced constants. Our risk bounds and sample complexity depend explicitly on the local patch radius $r$. 

\end{abstract}


\section{Introduction}

Quantum machine learning (QML) refers to a wide array of learning frameworks wherein quantum data types are used for learning tasks \cite{chang2025primerquantummachinelearning,11014055, acampora2025quantumcomputingartificialintelligence, arunachalam2017surveyquantumlearningtheory}. The field has seen significant development over recent years with the introduction of a plethora of novel learning frameworks and demonstrations.  Variational quantum algorithms (VQAs), namely parametrized quantum circuits (PQCs) \cite{ McClean_2016, PhysRevA.98.032309, dunjko2017machinelearningartificial, Wang_2024} have garnered significant attention within the QML due to their compatibility with near-term noisy intermediate-scale quantum (NISQ) devices \cite{Preskill_2018, RevModPhys.94.015004, gujju2024quantummachinelearningnearterm, Melnikov_2023}, albeit they still are constrained by phenomena such as barren plateaus \cite{Larocca_2025,cerezo2024doesprovableabsencebarren}, where gradients vanish exponentially with system size. There have been significant advancement in QML models for fault-tolerant application-scale quantum (FASQ) computers \cite{eisert2025mindgapsfraughtroad, preskill2025nisqmegaquopmachine, zimbors2025mythsquantumcomputationfault, Caro2022} as well.

However, a fundamental practical challenge lies in \emph{deploying} these models, especially during the inference phase. The computational cost of evaluating quantum models (both in NISQ and FASQ devices), in particular for gradient-based optimization, presents a significant bottleneck, as each gradient component typically requires multiple circuit evaluations \cite{molteni2025quantummachinelearningadvantages}. Unlike classical models, which, once trained, can be queried at negligible cost on classical hardware, quantum models impose a per-query cost in quantum hardware time, energy, and expense that scales with the circuit's complexity. As QML models scale over the years, this cost will scale alongside and translate over into higher energy and quantum hardware resources \cite{Tripathi:2025qki, fellousasiani2022resourcecostlargescale}. Classical surrogates of trained quantum learning models offer a potential pathway to mitigate these challenges. The core principle is to train on quantum hardware to extract potential advantages \cite{hangleiter2026quantumadvantageachieved,salmon2023provableadvantagequantumpac, huang2025vastworldquantumadvantage, eisert2025mindgapsfraughtroad, Schuld_2022} native to quantum domain learning and afterwards dequantize \cite{cotler2021revisitingdequantizationquantumadvantage, cerezo2024doesprovableabsencebarren} the inference phase of these expensive quantum models via constructing an equivalent classical counterpart of these trained quantum learning models. A particular type of PQCs named reuploading quantum models  \cite{P_rez_Salinas_2020, P_rez_Salinas_2021} wherein alternating layers of data-encoding gates and trainable unitary gates are used have been known to be exactly represented as truncated Fourier series \cite{PhysRevA.103.032430, PhysRevA.107.062612}. This property was used by Schreiber et al \cite{Schreiber_2023} to design an exact classical Fourier surrogates for this particular class of QML models. Hernich et al. \cite{hernicht2025enhancingscalabilityclassicalsurrogates} improved upon this work in terms of sample complexity by using random Fourier features instead of exact grid reconstruction. Jerbi et al \cite{Jerbi_2024} under a similar spirit introduced shadow models applicable to models beyond reuploading quantum models. In a similar vein, Watanabe et al. \cite{watanabe2026tensornetworksurrogatemodels} showed that two-dimensional tensor-network ansatzes can serve as accurate surrogate models for variational quantum circuits, with classical feasibility maintained at moderate bond dimension. These frameworks dealt with surrogating the entire data space of the trained quantum models and we will refer to these as global surrogation and the classical surrogates so constructed as global surrogates.

We focus on surrogate construction for a local patch inside the data space of an arbitrary trained quantum learning model. In previous work \cite{nair2025localsurrogatesquantummachine}, we numerically demonstrated that an arbitrary quantum learning model can be locally quantum surrogated with a reuploading model, within a local patch in the data space, and such a local quantum surrogate can always be further used to construct a local classical surrogate via the classical surrogation protocol outlined by Schreiber et al \cite{Schreiber_2023}. Furthermore, we provided empirical evidence for the cost optimization aspect of local surrogates, wherein smaller local patches required fewer resources. Local surrogates are widely featured in interpretable/explainable AI (XAI) literature \cite{molnar2025, Longo_2024, 10.1145/2939672.2939778}, and this has been introduced to the QML interpretability paradigm as well \cite{Pira_2024}. Lerch et al \cite{lerch2024efficient, mhiri2025unifyingaccountwarmstart} have introduced a hybrid quantum-classical expectation landscape simulation framework wherein local patches within these landscapes are efficiently classically simulated.

\begin{figure}[!t]
  \centering
  \includegraphics[width=\textwidth]{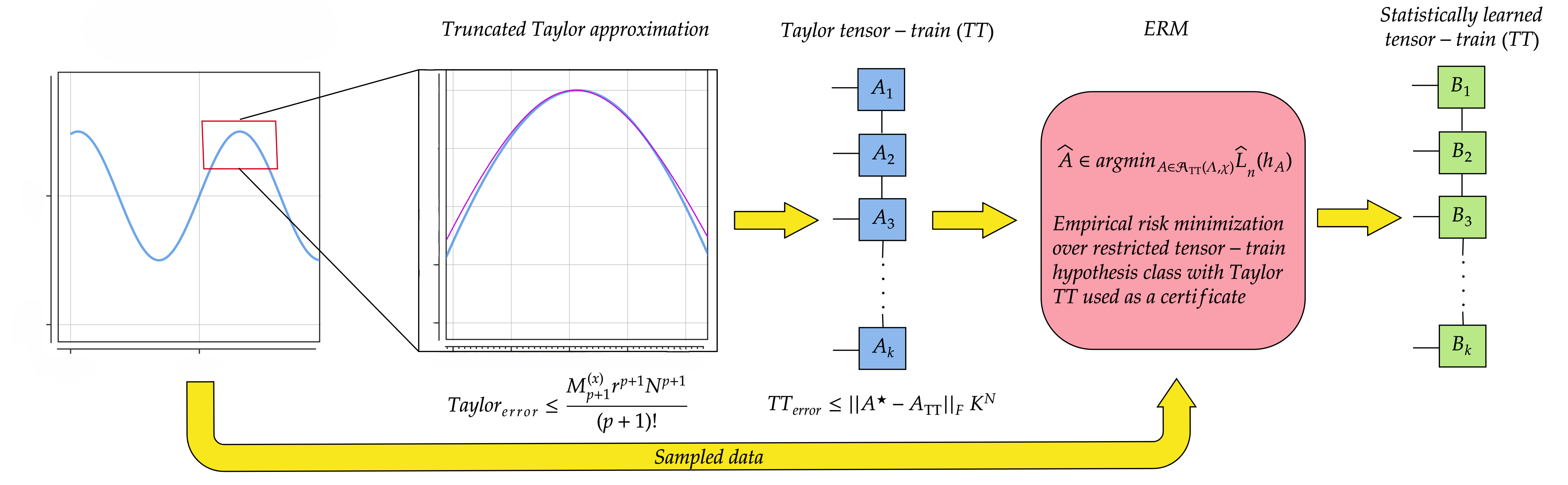}
  \caption{Sample illustration of the local tensor-train surrogates (LTTS):
  Starting from a function $f$, we restrict to a local patch and approximate $f$ by a
  truncated Taylor polynomial.  The Taylor coefficients are embedded as a tensor and represented in tensor-train (TT) format
  with cores $A_1,\dots,A_k$, which serves as an analytical
  certificate predictor with a known error bound decomposition.  An empirical risk
  minimizer (ERM) then fits a new TT with cores
  $B_1,\dots,B_k$ directly from sampled function values in the
  local patch that potentially uses the Taylor-TT certificate as a warm-start initialization instead of a random initialization. We present deterministic bounds on the certificate's approximation
  errors and statistical guarantees on the learned surrogate's generalization.}
  \label{fig:ltts_pipeline}
\end{figure}

\section*{Framework overview}

In this paper, we introduce local tensor-train (TT) surrogates (LTTS) for a broad class of trained quantum learning models \cite{Huggins_2019,_unkovi__2022,Rieser_2023}. Tensor trains (TT) \cite{Bridgeman_2017,eisert2013entanglementtensornetworkstates,Oseledets2011,perezgarcia2007matrixproductstaterepresentations}, also known in the physics literature as matrix product states (MPS), are a structured tensor network decomposition designed to efficiently represent quantum states with limited entanglement. We treat a trained quantum model as a black-box function and focus on a local patch around a point of interest. Our goal is to construct a local surrogate that is accurate in the patch, efficiently parameterized through the TT structure , and learnable from noisy samples with quantifiable sample complexity. We do not assume any special reconstruction structure intrinsic to the underlying quantum model. This distinguishes our setting from structure-dependent reconstruction frameworks such as those of Schreiber et al. \cite{Schreiber_2023}. Instead, our analysis combines classical approximation theory with probably approximately correct (PAC) learning guarantees \cite{ShalevShwartzBenDavid2014}.

The theoretical backbone of the framework is a deterministic error certificate constructed from a local Taylor approximation together with TT compression. The deterministic error terms originating from Taylor truncation are characterized using Taylor remainder terms, while the error term due to TT representation is characterized using the results of Oseledets \cite{Oseledets2011}. In statistical learning terms, this certificate plays the role of a witness hypothesis: it shows that the constrained TT class contains a good local approximant to the black-box function on the patch. Aldavero et al \cite{rodrguezaldavero2025chebyshevapproximationcompositionfunctions} had introduced Chebyshev approximation for quantum numerical analyses, wherein they traced out a similar polynomial approximation to TT embedding.

\medskip\noindent\textbf{Deterministic certificate}
(Theorem \ref{thm:deterministic}).
On a local patch of radius $r$, the Taylor--TT construction yields a predictor whose pointwise error over the patch is bounded by the sum of a Taylor truncation term (controlled by $r$ and $p$) and a TT compression term
(controlled by $\chi$):
\[
  |g(x) - h_{A_\chi}(x)|
  \;\le\;
  \underbrace{\varepsilon_{\mathrm{Taylor}}(r, p)}_{\text{Taylor truncation}}
  \;+\;
  \underbrace{\varepsilon_{TT}}_{\text{TT compression}}
  \;=:\; E_{\mathrm{det}}.
\]where $\varepsilon_{TT}$ absorbs the feature-norm constant. Both terms are deterministic approximation terms and do not depend on the number of training samples.

We set up a statistical learning task based on empirical risk minimization (ERM) \cite{788640,ShalevShwartzBenDavid2014}, wherein the goal is to learn a TT predictor from a constrained TT hypothesis (parameterized by TT cores and a fixed degree-p tensor-product feature map) class that closely approximates the target black-box function with quantifiable sample complexity. Using pseudo-dimension bounds for tensor-network models \cite{KhavariRabusseau2021} together with a uniform-convergence argument for bounded squared loss \cite{JMLR:v11:shalev-shwartz10a,Mohri2018,Vapnik1999-VAPTNO}, we obtain high-probability excess-risk guarantees for the learned surrogate. Thus, the deterministic certificate controls the approximation side of the problem, while the statistical analysis controls the estimation side. In practice, the learner does not need to compute Taylor coefficients or perform TT-SVD, although the certificate may also be useful as a warm start for the learning pipeline. In our analysis, the explicit dependence of the risk and sample-complexity bounds on the patch radius enters through this Taylor--TT construction. This yields polynomial complexity in the TT model parameters, although worst-case constants can still retain dimension dependence, thus inheriting the curse of dimensionality.

\medskip\noindent\textbf{Statistical guarantee}
(Theorem \ref{thm:end_to_end_pdim}).
When ERM is performed over the bounded TT hypothesis class, the population risk (expected squared error on unseen points) of the learned surrogate $\hat{h}$ satisfies, with probability at least $1 - \delta$,
\[
  R(\hat{h})
  \;\le\;
  E_{\mathrm{det}}^2
  \;+\;
  \Delta_n(\delta),
\]
where $\Delta_n(\delta) \to 0$ as $n \to \infty$.  The statistical term depends on the pseudo-dimension of the TT class, which scales as $d_{\mathrm{eff}} = O(N(p+1)\chi^2 \log N)$ i.e. polynomially in all relevant parameters, in contrast to the exponential $(p+1)^N$ scaling of an unconstrained polynomial model.  The certificate $E_{\mathrm{det}}^2$ enters as an upper bound on the best achievable risk in the class; ERM finds a hypothesis at most $\Delta_n$ worse.

\medskip\noindent\textbf{Local dependence on the patch radius}
(Theorem \ref{thm:local_ltts}).
Shrinking the patch radius $r$ directly improves the Taylor truncation component of the deterministic certificate and reduces the norm budget $\Lambda^\star(r)$ entering the statistical bound.  The overall benefit also depends on how the embedded Taylor tensor $A_\star$ compresses under TT-SVD, which is problem-dependent and can be assessed empirically.  Our numerical experiments indicate that, for the quantum models tested, smaller patches are associated with lower approximation error.

\begin{figure}[h]
\centering
\fbox{\parbox{0.92\textwidth}{
\textbf{Workflow:} Local TT surrogate learning on a patch $\cB(x_0,r)$\\[4pt]
\textbf{Input:} Trained quantum model $g$, patch center $x_0$, radius $r$,
degree $p$, TT rank $\chi$\\
\textbf{Output:} Learned local classical surrogate $\hat h$ on $\cB(x_0,r)$\\[6pt]
1.\quad Construct a local Taylor certificate around $x_0$
(optionally estimating derivatives in practice)\\
2.\quad Embed the local coefficient tensor into TT format 
with rank cap $\chi$ to obtain a certificate predictor\\
3.\quad Sample $n$ points in $\cB(x_0,r)$ and query $g$ to obtain training
pairs $(X_i,Y_i)$\\
4.\quad Learn a TT surrogate by empirical risk minimization over a constrained
TT class (optionally warm-started from Step 2)\\
5.\quad Deploy the learned surrogate classically on the local patch\\[4pt]
{\small Steps 1--2 provide the deterministic certificate and an optional warm start.
Step 4 is the statistical learning stage. The full theory combines the certificate
with the ERM guarantee to obtain an explicit local risk bound.}
}}
\end{figure}

\section{Truncated Taylor Model}
We will begin with formally defining the local patch we intend to surrogate.
\begin{definition}[Patch and normalized coordinates]\label{def:patch}
Fix a center $x_0 \in \R^N$ and radius $r > 0$. Define the $\ell^\infty$-hypercube
\[
\cB(x_0,r) := \{x \in \R^N : \|x - x_0\|_\infty \le r\}.
\]
For $x \in \cB(x_0,r)$, define the normalized coordinate $\xi \in [-1,1]^N$ by
\[
\xi := \frac{x - x_0}{r} \qquad \text{(componentwise)}.
\]
\end{definition}

Local tensor-train surrogates (LTTS) are formalized for a locally smooth blackbox function, and we formally define these as follows:

\begin{definition}[Local function]\label{ass:smooth}

Let $g:\cB(x_0,r)\to\mathbb R$ be the target function, assumed to be $(p+1)$-times
continuously differentiable on $\cB(x_0,r)$. Define the normalized local function
$f:[-1,1]^N\to\mathbb R$ by
\[
f(\xi):=g(x_0+r\xi).
\]
Also define

 \[
C_{p+1}^{(x)} := \sup_{x \in \cB(x_0,r)} \max_{|\alpha| = p+1} |\partial_x^\alpha g(x)| < \infty \]
and there exists a finite constant:
\[
C_{\le p}^{(x)} := \sup_{x \in \cB(x_0,r)} \max_{0 \le |\alpha| \le p} |\partial_x^\alpha g(x)| < \infty.
\]
    
\end{definition}

The local patch function is approximated with a truncated Taylor series, and we present the accrued truncation error bound in this section. We use standard multi-index notation throughout. For $\alpha = (\alpha_1, \ldots, \alpha_N) \in \Nzero^N$, let $|\alpha| = \sum_i \alpha_i$, $\alpha! = \prod_i \alpha_i!$, and $\xi^\alpha = \prod_i \xi_i^{\alpha_i}$.

\begin{theorem}[Taylor truncation error]\label{thm:taylor}
Let $g:\R^N \to \R$ be defined as in Definition \ref{ass:smooth}. Define the total-degree-$p$ Taylor polynomial on the normalized patch:
\[
T_p(\xi) := \sum_{|\alpha| \le p} \frac{r^{|\alpha|}}{\alpha!} (\partial^\alpha g)(x_0) \, \xi^\alpha, \qquad \xi \in [-1,1]^N.
\]
Then for all $\xi \in [-1,1]^N$,
\[
\bigl| g(x_0 + r\xi) - T_p(\xi) \bigr| \le \frac{C_{p+1}^{(x)} \, r^{p+1} \, N^{p+1}}{(p+1)!}.
\]
\end{theorem}

\begin{proof}
For a scalar function $g:\R^N \to \R$ that is $(p+1)$-times
continuously differentiable on $\cB(x_0,r)$, the Taylor expansion around $x_0$ up to total degree $p$ is
\[
g(x) = \sum_{|\alpha| \le p} \frac{1}{\alpha!} (\partial^\alpha g)(x_0) \cdot (x - x_0)^\alpha + R_p(x),
\]
where $R_p(x)$ is the remainder term. The standard multivariate remainder bound gives
\[
|R_p(x)| \le \sup_{z \in [x_0, x]} \max_{|\alpha| = p+1} |\partial^\alpha g(z)| \cdot \sum_{|\alpha| = p+1} \frac{|x - x_0|^\alpha}{\alpha!},
\]
where $[x_0, x]$ denotes the line segment from $x_0$ to $x$.

Set $x = x_0 + r\xi$ where $\xi \in [-1,1]^N$. Then $x - x_0 = r\xi$ and
\[
|x - x_0|^\alpha = (r|\xi|)^\alpha = r^{|\alpha|} |\xi|^\alpha.
\]
Since $|\xi_i| \le 1$, we have $|\xi|^\alpha \le 1$. For $|\alpha| = p+1$:
\[
|R_p(x_0 + r\xi)| \le C_{p+1}^{(x)} \, r^{p+1} \sum_{|\alpha| = p+1} \frac{1}{\alpha!}.
\]
Using the multinomial identity $\sum_{|\alpha| = m} \frac{1}{\alpha!} = \frac{N^m}{m!}$ with $m = p+1$ yields
\[
|R_p(x_0 + r\xi)| \le \frac{C_{p+1}^{(x)} \, r^{p+1} \, N^{p+1}}{(p+1)!}. \qedhere
\]
\end{proof}

\section{Embedding Scheme}

In order to embed the truncated Taylor polynomial constructed from the local patch into a tensor-train (TT) format, we define a custom TT-compatible feature map wherein the factorial terms of the Taylor polynomial form a factorial feature vector as defined below:

\begin{definition}[Factorial feature map of degree $p$]\label{def:features}
Fix an integer $p \ge 0$. For $\xi_i \in [-1,1]$, define the factorial feature vector
\[
v(\xi_i) \in \R^{p+1}, \qquad v(\xi_i)_k := \frac{\xi_i^k}{k!}, \quad k = 0, 1, \ldots, p.
\]
Define the tensor-product feature map
\[
\Phi(\xi) := v(\xi_1) \otimes v(\xi_2) \otimes \cdots \otimes v(\xi_N) \in \R^{(p+1)^N}.
\]
\end{definition}

The feature map $\Phi(\xi)$ is naturally indexed by multi-indices $\alpha \in \{0, 1, \ldots, p\}^N$ (the ``box'' truncation), with
\[
\Phi(\xi)[\alpha] = \prod_{i=1}^{N} \frac{\xi_i^{\alpha_i}}{\alpha_i!}.
\]

Throughout this paper, we use the following terminology:
\begin{itemize}
\item The \emph{simplex} (or \emph{total-degree index set}) refers to 
      $\{\alpha \in \mathbb{N}_0^N : |\alpha| \leq p\}$, which indexes 
      the Taylor polynomial $T_p$.
\item The \emph{box} (or \emph{product index set}) refers to 
      $\{0,1,\ldots,p\}^N$, the Cartesian product structure required 
      for tensor-train decomposition.
\end{itemize}

The Taylor polynomial $T_p(\xi)$ in Theorem \ref{thm:taylor} uses \emph{simplex truncation} (total degree $|\alpha| \le p$) which has $\binom{N+p}{p}$ terms, whereas the TT-compatible tensor-product feature map $\Phi(\xi)$ is indexed by the \emph{box} $\{0, \ldots, p\}^N$ which has $(p+1)^N$ entries. To express $T_p$ as an inner product with $\Phi$, we embed the simplex coefficients into the box tensor space by zero-padding. The zero-padding ensures that the ``extra'' box entries (those with $|\alpha| > p$) do not contribute to the inner product. We provide numerical evidence substantiating that such a zero-padding embedding scheme doesn't exhibit systematic rank inflation under TT decomposition. This embedding is what enables TT approximation of the corresponding coefficient tensor as defined below, while preserving the Taylor polynomial structure. 

\begin{definition}[Embedded Taylor coefficient tensor]\label{def:Astar}
Define the order-$N$ tensor $A^\star \in \R^{(p+1) \times \cdots \times (p+1)}$ indexed by $\alpha \in \{0, 1, \ldots, p\}^N$ via
\[
A^\star[\alpha] := 
\begin{cases}
r^{|\alpha|} (\partial^\alpha g)(x_0), & \text{if } |\alpha| \le p, \\[4pt]
0, & \text{if } |\alpha| > p.
\end{cases}
\]
Equivalently, using the normalized target $f(\xi) = g(x_0 + r\xi)$ and the chain rule identity $\partial_\xi^\alpha f(0) = r^{|\alpha|} (\partial_x^\alpha g)(x_0)$:
\[
A^\star[\alpha] = 
\begin{cases}
\partial_\xi^\alpha f(0), & \text{if } |\alpha| \le p, \\[4pt]
0, & \text{if } |\alpha| > p.
\end{cases}
\]
Note that $A^\star$ stores the \emph{unnormalized} derivatives; the factorials are carried by the feature map.
\end{definition}

\begin{lemma}[Taylor polynomial as inner product]\label{lem:taylor_ip}
With $A^\star$ as in Definition \ref{def:Astar} and $\Phi$ as in Definition \ref{def:features}, for all $\xi \in [-1,1]^N$:
\[
T_p(\xi) = \ip{A^\star}{\Phi(\xi)},
\]
where $\ip{\cdot}{\cdot}$ denotes the Frobenius inner product.
\end{lemma}

\begin{proof}
Expanding the inner product:
\begin{align*}
\ip{A^\star}{\Phi(\xi)} 
&= \sum_{\alpha \in \{0,\ldots,p\}^N} A^\star[\alpha] \, \Phi(\xi)[\alpha] \\
&= \sum_{\alpha \in \{0,\ldots,p\}^N} A^\star[\alpha] \prod_{i=1}^{N} \frac{\xi_i^{\alpha_i}}{\alpha_i!} \\
&= \sum_{|\alpha| \le p} r^{|\alpha|} (\partial^\alpha g)(x_0) \prod_{i=1}^{N} \frac{\xi_i^{\alpha_i}}{\alpha_i!} \\
&= \sum_{|\alpha| \le p} \frac{r^{|\alpha|}}{\alpha!} (\partial^\alpha g)(x_0) \, \xi^\alpha \\
&= T_p(\xi),
\end{align*}
where the third equality uses that $A^\star[\alpha] = 0$ for $|\alpha| > p$ as per zero-padding, and the fourth uses $\alpha! = \prod_i \alpha_i!$.
\end{proof}

We explicitly bound the feature vector defined in Definition \ref{def:features} and identified the bound as a modified Bessel function. Later, this bound will feature in the TT error term.

\begin{lemma}[Bessel feature norm bound]\label{lem:bessel}
Let $I_0$ denote the modified Bessel function of the first kind (order $0$):
\[
I_0(2) = \sum_{k=0}^{\infty} \frac{1}{(k!)^2}.
\]
Define $K := \sqrt{I_0(2)} \approx 1.50983$ (numerically). Then for all $\xi \in [-1,1]^N$,
\[
\|\Phi(\xi)\|_2 \le K^N.
\]
\end{lemma}

\begin{proof}
By definition of $\Phi(\xi)$ as a tensor (Kronecker) product and using $\|u \otimes w\|_2 = \|u\|_2 \|w\|_2$ we have:
\[
\|\Phi(\xi)\|_2 = \|v(\xi_1) \otimes \cdots \otimes v(\xi_N)\|_2 = \prod_{i=1}^{N} \|v(\xi_i)\|_2,
\]

For each coordinate $i$, since $|\xi_i| \le 1$:
\[
\|v(\xi_i)\|_2^2 = \sum_{k=0}^{p} \left( \frac{\xi_i^k}{k!} \right)^2 \le \sum_{k=0}^{\infty} \frac{1}{(k!)^2} = I_0(2).
\]
Thus $\|v(\xi_i)\|_2 \le \sqrt{I_0(2)} = K$ for each $i$, and multiplying over $i = 1, \ldots, N$ gives $\|\Phi(\xi)\|_2 \le K^N$.
\end{proof}


 $K^N$ grows exponentially in $N$. This exponential factor is a manifestation of the curse of dimensionality for full tensor-product features.

\begin{lemma}[$r$-dependent upper bound on $\|A^\star\|_F$]\label{lem:lambda_star}
Under Definition \ref{ass:smooth}, the embedded Taylor coefficient tensor satisfies
\[
\|A^\star\|_F \le \Lambda^\star(r) := C_{\le p}^{(x)} \left( \sum_{m=0}^{p} \binom{N + m - 1}{m} r^{2m} \right)^{1/2}.
\]
\end{lemma}

\begin{proof}
For each multi-index $\alpha$ with $|\alpha| = m \le p$ by Definition \ref{def:Astar}:
\[
|A^\star[\alpha]| = |r^m (\partial^\alpha g)(x_0)| \le r^m C_{\le p}^{(x)}.
\]
For $|\alpha| > p$, we have $A^\star[\alpha] = 0$ by same definition. Therefore:
\[
\|A^\star\|_F^2 = \sum_{\alpha \in \{0,\ldots,p\}^N} |A^\star[\alpha]|^2 = \sum_{m=0}^{p} \sum_{\substack{\alpha \in \Nzero^N \\ |\alpha| = m}} |A^\star[\alpha]|^2 \le (C_{\le p}^{(x)})^2 \sum_{m=0}^{p} r^{2m} \cdot \#\{\alpha \in \Nzero^N : |\alpha| = m\}.
\]
where $\#$ represents the cardinality of the set. The number of nonnegative integer solutions to $\alpha_1 + \cdots + \alpha_N = m$ is $\binom{N + m - 1}{m}$. Taking the square root yields the claim.
\end{proof}

\section{Tensor-Train (TT) Representation}
We formalize a tensor-train using the following definition:


\begin{definition}[TT-rank and TT tensor class]\label{def:tt}
An order-$N$ tensor $A \in \R^{(p+1) \times \cdots \times (p+1)}$ has TT-rank at most $\chi$ if it admits a representation
\[
A[\alpha_1, \ldots, \alpha_N] = G_1[\alpha_1] \, G_2[\alpha_2] \cdots G_N[\alpha_N],
\]
where $G_k[\alpha_k] \in \R^{r_{k-1} \times r_k}$ are matrices (cores) with $r_0 = r_N = 1$ and $r_k \le \chi$ for $k = 1, \ldots, N-1$.

For $\Lambda > 0$ and $\chi \in \N$, define the constrained tensor class
\[
\cA_{\mathrm{TT}}(\Lambda, \chi) := \{ A \in \R^{(p+1) \times \cdots \times (p+1)} : \rankTT(A) \le \chi, \, \|A\|_F \le \Lambda \}.
\]
\end{definition}

\begin{definition}[Best TT approximation error in the constrained class]\label{def:epsTT}
Fix $\Lambda>0$ and a TT-rank budget $\chi\in\N$. Define
\[
\varepsilon_{\mathrm{TT}}(\Lambda,\chi) := \inf \{ \|A^\star - A\|_F : A \in \cA_{\mathrm{TT}}(\Lambda,\chi)\}.
\]
\end{definition}

\begin{lemma}[TT-SVD error bound]\label{lem:ttsvd}
Let $A$ be an order-$N$ tensor and let $A_{\mathrm{TT}}$ be obtained by TT-SVD (left-orthonormal / left-canonical) with uniform rank cap $\chi$. Let $\sigma_{k,j}$ denote the singular values of the $k$-th unfolding of $A$. Then
\[
\|A - A_{\mathrm{TT}}\|_F^2 \le \sum_{k=1}^{N-1} \sum_{j > \chi} \sigma_{k,j}^2.
\]
\end{lemma}

\begin{proof}
This is the standard TT-SVD quasi-optimality estimate; see Oseledets \cite{Oseledets2011}, Eq. (2.4).
\end{proof}

\section{Deterministic Error Bounds}
We have all the pieces to put together a deterministic error bound for the Taylor-TT certificate predictor.


\begin{lemma}[Relation between $x$- and $\xi$-derivative bounds]\label{lem:scale_derivs}
As defined under the Definition \ref{ass:smooth}, the normalized target $f(\xi) = g(x_0 + r\xi)$ satisfies
\[
C_{p+1}^{(\xi)} := \sup_{\xi \in [-1,1]^N} \max_{|\alpha| = p+1} |\partial_\xi^\alpha f(\xi)| \le r^{p+1} C_{p+1}^{(x)}.
\]
\end{lemma}

\begin{proof}
By the chain rule, $\partial_{\xi_i} = r \, \partial_{x_i}$. Thus $\partial_\xi^\alpha f(\xi) = r^{|\alpha|} (\partial_x^\alpha g)(x_0 + r\xi)$. Taking $|\alpha| = p+1$ and suprema yields the result.
\end{proof}


\begin{lemma}[TT approximation error propagates with $K^N$]\label{lem:tt_comp}

Let $\xi=(x-x_0)/r$. Then for all $x\in\cB(x_0,r)$,
\[
|T_p(\xi) - h_{A_{\mathrm{TT}}}(x)| \le \|A^\star-A_{\mathrm{TT}}\|_F K^N.
\]

\end{lemma}

\begin{proof}
Recall that $\xi=(x-x_0)/r$. By Lemma \ref{lem:taylor_ip},
\[
T_p(\xi)=\ip{A^\star}{\Phi(\xi)}=h_{A^\star}(x).
\]
Apply Lemma \ref{lem:predbound} with $A = A^\star$ and $B = A_{\mathrm{TT}}$.
\end{proof}

\begin{theorem}[Deterministic error decomposition (Taylor-TT certificate predictor)]\label{thm:deterministic}
Let $g$ be defined under Definition \ref{ass:smooth}. Let $A_{\mathrm{TT}}\in \cA_{\mathrm{TT}}(\Lambda,\chi)$ be arbitrary.
Then for all $x\in\cB(x_0,r)$, writing $\xi=(x-x_0)/r$, we have
\[
\bigl| g(x) - h_{A_{\mathrm{TT}}}(x) \bigr|
\le
\underbrace{\frac{C_{p+1}^{(x)} \, r^{p+1} \, N^{p+1}}{(p+1)!}}_{\text{Taylor truncation}}
\;+\;
\underbrace{\|A^\star-A_{\mathrm{TT}}\|_F \, K^N}_{\text{TT approximation}}.
\]
\end{theorem}

\begin{proof}
By the triangle inequality,
\[
|g(x)-h_{A_{\mathrm{TT}}}(x)|
\le |g(x_0+r\xi)-T_p(\xi)| + |T_p(\xi)-h_{A_{\mathrm{TT}}}(x)|.
\]
Apply Theorem \ref{thm:taylor} to the first term and Lemma \ref{lem:tt_comp} to the second.
\end{proof}

\begin{remark}[When the deterministic certificate is informative]\label{rem:informative}
Theorem \ref{thm:deterministic} holds for any $(\Lambda,\chi)$ and does \emph{not} assume
that $A^\star$ is well-compressible in TT format. The bound is quantitatively informative
when the TT approximation term does not dominate over the Taylor truncation term. This condition is problem-dependent and can be assessed empirically by measuring the TT
approximation error of the embedded Taylor tensor $A^\star$ as a function of $\chi$.
\end{remark}

\section{Statistical Learning Framework}\label{sec:statistical}

Now that we have a certificate predictor with deterministic error bounds, we will embed it into a statistical learning framework. For this, we define an appropriate hypothesis class for the learning task.

\begin{definition}[TT hypothesis class]\label{def:hypclass}
Define the hypothesis class of TT-parameterized predictors:
\[
\cH_{\mathrm{TT}}(\Lambda, \chi) := \left\{ h_A(x) = \ip{A}{\Phi\left( \frac{x - x_0}{r} \right)} : A \in \cA_{\mathrm{TT}}(\Lambda, \chi) \right\}.
\]
Equivalently, $h_A(x) = \tilde{h}_A(\xi)$ where $\tilde{h}_A(\xi) := \ip{A}{\Phi(\xi)}$ and $\xi = (x - x_0)/r$.
\end{definition}

$\mathcal{H}_{\mathrm{TT}}(\Lambda,\chi)$ in Definition \ref{def:hypclass} is \emph{not} restricted to simplex
support: an arbitrary $A\in\mathcal{A}_{\mathrm{TT}}(\Lambda,\chi)$ may place mass on $|\alpha|>p$ terms.
Our deterministic Taylor-TT bound (Theorem \ref{thm:deterministic}) should be interpreted as an
\emph{approximation certificate} for the class: it upper-bounds the error of the specific predictor
$h_{A_{\mathrm{TT}}}$ obtained by TT-approximating the embedded Taylor tensor $A^\star$.

\begin{lemma}[Prediction boundedness and Lipschitz continuity in $A$]\label{lem:predbound}
For any tensors $A, B \in \R^{(p+1) \times \cdots \times (p+1)}$ and any $x \in \cB(x_0, r)$:
\[
|h_A(x)| \le \|A\|_F \, K^N, \qquad |h_A(x) - h_B(x)| \le \|A - B\|_F \, K^N.
\]
For $A \in \cA_{\mathrm{TT}}(\Lambda, \chi)$, we have $|h_A(x)| \le \Lambda K^N$ uniformly over $x \in \cB(x_0, r)$.
\end{lemma}

\begin{proof}
Recall that $\xi = (x - x_0)/r \in [-1,1]^N$. By Cauchy-Schwarz inequality and Lemma \ref{lem:bessel}:
\[
|h_A(x)| = |\ip{A}{\Phi(\xi)}| \le \|A\|_F \|\Phi(\xi)\|_2 \le \|A\|_F \, K^N.
\]
The Lipschitz bound follows similarly: $|h_A(x) - h_B(x)| = |\ip{A - B}{\Phi(\xi)}| \le \|A - B\|_F \, K^N$.
\end{proof}

\begin{definition}[Local TT hypothesis class]\label{def:localhypclass}
Given patch radius $r > 0$ and bond dimension $\chi \in \mathbb{N}$, define the canonical local hypothesis class $\cH_{\mathrm{TT}}(\Lambda^\star(r), \chi)$ where $\Lambda^\star(r)$ is given by Lemma \ref{lem:lambda_star}. By construction, $A^\star \in \cA_{\mathrm{TT}}(\Lambda^\star(r), \chi)$ whenever $\rankTT(A^\star) \le \chi$.

\end{definition}

With the existence of a good hypothesis established via the Taylor-TT certificate with deterministic error decomposition in Theorem \ref{thm:deterministic}, we define a \emph{learning task} as a regression problem, where the learner collects input-output samples $(X_i, Y_i)$ from the patch (where $Y_i = g(X_i) + \varepsilon_i$ may include measurement noise) and fits a predictor from a constrained tensor-train (TT) hypothesis class that minimizes the expected squared loss via empirical risk minimization (ERM) \cite{788640,ShalevShwartzBenDavid2014}. The statistical analysis proceeds by using standard theoretical machine learning tools to bound the population/true risk.  We use
uniform convergence \cite{JMLR:v11:shalev-shwartz10a} of empirical risk to population risk to show that the ERM solution generalizes i.e. low
training error implies low generalization error. We will now elaborate on these steps in the following sections.

\begin{definition}[Local regression model]\label{ass:regression}
We define \emph{local regression model} on the patch $\cB(x_0, r)$ as a tuple
$(P_X,\, g,\, \varepsilon,\, G_{\max},\, \sigma)$ where:
\begin{enumerate}
  \item $(X_i, Y_i)_{i=1}^n$ are i.i.d. with $X_i \sim P_X$ (probability distribution)
        supported on $\cB(x_0, r)$;
  \item $Y_i = g(X_i) + \varepsilon_i$ with
        $\E[\varepsilon_i \mid X_i] = 0$; (unbiased error/noise)
  \item $|g(x)| \le G_{\max}$ for all $x \in \cB(x_0, r)$
        and $|\varepsilon_i| \le \sigma$ a.s. (bounded signal and noise)
\end{enumerate}
We write $|Y_i| \le Y_{\max} := G_{\max} + \sigma$ a.s.
\end{definition}

\begin{definition}\label{def:risk}
    Define the population and empirical risks:
\[
L(h) := \E[(h(X) - Y)^2], \qquad \hat{L}_n(h) := \frac{1}{n} \sum_{i=1}^{n} (h(X_i) - Y_i)^2.
\]
Define the clean (noise-free) risk:
\[
R(h) := \E[(h(X) - g(X))^2].
\]
\end{definition}

\begin{lemma}[Excess risk transfers from $L$ to $R$]\label{lem:cleanrisk}
Under $\E[\varepsilon \mid X] = 0$:
\[
L(h) = R(h) + \E[\varepsilon^2] \qquad \text{for all measurable } h,
\]
so for any hypothesis class $\cH$:
\[
L(h) - \inf_{h' \in \cH} L(h') = R(h) - \inf_{h' \in \cH} R(h').
\]
\end{lemma}

\begin{proof}
Expand $L(h) = \E[(h - g - \varepsilon)^2] = \E[(h-g)^2] + \E[\varepsilon^2] - 2\E[(h-g)\varepsilon]$. Using the tower property and $\E[\varepsilon \mid X] = 0$:
\[
\E[(h-g)\varepsilon] = \E[\E[(h(X)-g(X))\varepsilon \mid X]] = \E[(h(X)-g(X)) \cdot \E[\varepsilon \mid X]] = 0. \qedhere
\]
\end{proof}


\begin{lemma}[Compactness of the TT constraint set]\label{lem:compact}
For fixed $(N, p, \chi, \Lambda)$, the set $\cA_{\mathrm{TT}}(\Lambda, \chi) \subset \R^{(p+1)^N}$ is compact. Consequently, there exists an empirical risk minimizer (ERM):
\[
\hat{A} \in \argmin_{A \in \cA_{\mathrm{TT}}(\Lambda, \chi)} \hat{L}_n(h_A), \qquad \hat{h} := h_{\hat{A}} \in \argmin_{h \in \cH_{\mathrm{TT}}(\Lambda, \chi)} \hat{L}_n(h).
\]
\end{lemma}

\begin{proof}
The Frobenius ball $\{A : \|A\|_F \le \Lambda\}$ is compact in finite-dimensional Euclidean space (closed and bounded). 

The set $\{A : \rankTT(A) \le \chi\}$ is closed: the TT-rank of $A$ equals $\max_{k=1,\ldots,N-1} \mathrm{rank}(A^{\langle k \rangle})$, where $A^{\langle k \rangle} \in \R^{(p+1)^k \times (p+1)^{N-k}}$ is the $k$-th unfolding (matricization) of $A$. The map $A \mapsto A^{\langle k \rangle}$ is linear (hence continuous), and the set of matrices with rank $\le \chi$ is closed in the Frobenius topology \cite[Sec. 2]{HiriartUrrutyLe2013rank}.

Therefore $\cA_{\mathrm{TT}}(\Lambda, \chi)$ is a closed subset
of a compact set, and is therefore compact.

Continuity of $\hat{L}_n(h_A)$ in $A$ follows from continuity of $A \mapsto h_A(x) = \ip{A}{\Phi((x-x_0)/r)}$ for each $x$. A continuous function on a compact set attains its minimum.
\end{proof}

\section{PAC Bounds via Pseudo-Dimension}

\begin{definition}[Bounded loss class] \label{bounded_loss_class}
Let $\cH_{\mathrm{TT}}(\Lambda,\chi)$ be as in Definition \ref{def:hypclass}.
From Lemma \ref{lem:predbound} we have $|h_A(x)| \le \Lambda K^N$. 
Under Definition \ref{ass:regression}, $|Y|\le Y_{\max}$ almost surely, 
and $|h(X)|\le \Lambda K^N$ almost surely for all $h\in\cH_{\mathrm{TT}}(\Lambda,\chi)$. 
Therefore the squared loss $\ell_h(x,y) := (h(x)-y)^2$ is uniformly bounded:
\[
0 \le \ell_h(X,Y) \le M(\Lambda) := \bigl(\Lambda K^N + Y_{\max}\bigr)^2 \quad \text{a.s.}
\]
In the local setting, substituting the canonical choice 
$\Lambda = \Lambda^\star(r)$ from Definition \ref{def:localhypclass} gives
\[
M(r) := \bigl(\Lambda^\star(r)\, K^N + Y_{\max}\bigr)^2,
\]
which decreases as $r \to 0$ since $\Lambda^\star(r) \to C_{\le p}^{(x)}$.    
\end{definition}

We write $h_A(x)=\langle A,\Phi((x-x_0)/r)\rangle$. Although $\Phi$ lives in an exponentially
large feature space, the TT constraint restricts $A$ to a low-parameter tensor network family. We use pseudo-dimension to quantify complexity of the hypothesis class which is required to prove the uniform convergence of the population and empirical risk classes. For this, we use a standard pseudo-dimension bound for tensor network regression models in terms of the number of real parameters in the network \cite{KhavariRabusseau2021}.

\begin{theorem}[TN pseudo-dimension via parameter count {\cite[Thm. 2]{KhavariRabusseau2021}}]\label{thm:tn_pdim}
Let $G=(V,E,\mathrm{dim})$ be a tensor-network structure, and let $\cH^{\mathrm{regression}}_G$
be the associated real-valued regression hypothesis class parameterized by $N_G$ real parameters.
Then
\[
\mathrm{Pdim}\!\left(\cH^{\mathrm{regression}}_G\right)\le 2\,N_G\,\log(12|V|).
\]
\end{theorem}

\begin{corollary}[TT pseudo-dimension (via restriction of a TN model)]\label{cor:tt_pdim}
Consider TT tensors of order $N$ and mode size $m:=p+1$ with TT ranks $(r_k)_{k=1}^{N-1}$,
where $r_0=r_N=1$. Parameterize such tensors by TT cores $G_k\in\R^{r_{k-1}\times m\times r_k}$.
The total number of real core parameters is
\[
N_G=\sum_{k=1}^{N} m\,r_{k-1}r_k.
\]
$\cH^{\mathrm{regression}}_{\mathrm{TT}}$ denotes the associated (unconstrained) TT regression class
$x\mapsto \langle A(G),\Phi((x-x_0)/r)\rangle$ induced by these parameters. By
Theorem \ref{thm:tn_pdim},
\[
\mathrm{Pdim}\!\left(\cH^{\mathrm{regression}}_{\mathrm{TT}}\right)
\le 2\,N_G\,\log(12N).
\]
By Def \ref{def:tt}, $\cH_{\mathrm{TT}}(\Lambda,\chi)$ is a \emph{restriction} of this larger class (imposing the constraints $\|A\|_F\le \Lambda$ and $r_k\le \chi$), pseudo-dimension cannot increase under
restriction, hence
\[
\mathrm{Pdim}\!\left(\cH_{\mathrm{TT}}(\Lambda,\chi)\right)
\le \mathrm{Pdim}\!\left(\cH^{\mathrm{regression}}_{\mathrm{TT}}\right)
\le 2\Big(\sum_{k=1}^{N} m\,r_{k-1}r_k\Big)\log(12N)
\le 2\,N\,m\,\chi^2\,\log(12N),
\]
where the last inequality uses $r_k\le\chi$ for $k=1,\dots,N-1$ (and $\chi\ge 1$).
\end{corollary}


The uniform convergence theorem we apply later on is stated in terms of the pseudo-dimension of the
\emph{loss class} and not the \emph{predictor class} as in corollary \ref{cor:tt_pdim}. 
Define the squared-loss class
\[
\cG := \left\{(x,y)\mapsto (h(x)-y)^2 : h\in\cH_{\mathrm{TT}}(\Lambda,\chi)\right\}.
\]

\begin{lemma}[Squared-loss lifting: $\mathrm{Pdim}(\cG)$ from $\mathrm{Pdim}(\cH)$]\label{lem:loss_pdim_sq}
Assume $|Y|\le Y_{\max}$ a.s.\ and $|h(X)|\le B$ a.s.\ for all $h\in\cH$.
Let $\cG_\cH:=\{(x,y)\mapsto (h(x)-y)^2 : h\in\cH\}$.
Then
\[
\mathrm{Pdim}(\cG_\cH)\;\le\; 4(2\,\mathrm{Pdim}(\cH) + 1)\log_2(6).
\]
Writing $d_H:=\mathrm{Pdim}(\cH_{\mathrm{TT}}(\Lambda,\chi))$, we may define
\[
d:=\mathrm{Pdim}(\cG)\;\le\; 4(2d_H+1)\log_2(6).
\]
\end{lemma}

\begin{proof}
Let $d_H=\mathrm{Pdim}(\cH)$. By the definition of pseudo-dimension,
$\mathrm{Pdim}(\cG_\cH)$ is the VC-dimension of the thresholded concept class
\[
\mathcal{C}
:=
\Big\{(x,y,t)\mapsto \mathbf{1}\!\big((h(x)-y)^2>t\big)\;:\;h\in\cH\Big\}.
\]
If $t<0$ then $(h(x)-y)^2>t$ holds trivially and such coordinates cannot increase VC-dimension.
Thus we may restrict to $t\ge 0$ without loss.
Write $\tilde{t}:=\sqrt{t}\ge 0$. Then
\[
(h(x)-y)^2>t
\quad\Longleftrightarrow\quad
\big(h(x)>y+\tilde{t}\big)\ \ \text{or}\ \ \big(h(x)<y-\tilde{t}\big).
\]
Each concept in $\mathcal{C}$ corresponds to the set union $S_h^+\cup S_h^-$, where
\[
S_h^+ := \big\{(x,y,\tilde{t}): h(x)>y+\tilde{t}\big\},
\qquad
S_h^- := \big\{(x,y,\tilde{t}): h(x)<y-\tilde{t}\big\}.
\]
Define the set families $\mathcal{F}_+:=\{S_h^+: h\in\cH\}$ and $\mathcal{F}_-:=\{S_h^-: h\in\cH\}$.
By restriction arguments analogous to those for the standard threshold class
\[
\mathcal{T}:=\{(x,t)\mapsto \mathbf{1}[h(x)>t]: h\in\cH\},
\]
we have $\mathrm{VCdim}(\mathcal{F}_+)\le d_H$ and $\mathrm{VCdim}(\mathcal{F}_-)\le d_H$.
Define $\mathcal{F}:=\mathcal{F}_+\cup\mathcal{F}_-$. By a standard union-growth argument (We give full proof of this inequality in the appendix as Lemma \ref{lem:vc-union}),
\[
\mathrm{VCdim}(\mathcal{F}) \le \mathrm{VCdim}(\mathcal{F}_+) + \mathrm{VCdim}(\mathcal{F}_-) + 1 \le 2d_H+1.
\]
Since $\mathcal{C}\subseteq \mathcal{F}^{\cup 2}:=\{A\cup B: A,B\in\mathcal{F}\}$,
Lemma 3.2.3 of \cite{BlumerEtAl1989} (applied with union parameter $s=2$) yields
\[
\mathrm{VCdim}(\mathcal{C})
< 2 \cdot \mathrm{VCdim}(\mathcal{F}) \cdot 2 \cdot \log_2(6)
\le 4(2d_H+1)\log_2(6),
\]
where we used that $\mathrm{VCdim}(\mathcal{C})$ is an integer in the last step.
Thus, 
\[
\mathrm{Pdim}(\cG_\cH)=\mathrm{VCdim}(\mathcal{C})\le 4(2d_H+1)\log_2(6)
\]
\end{proof}

\begin{remark}[Explicit TT loss-class bound]
Combining Lemma \ref{lem:loss_pdim_sq} with Corollary \ref{cor:tt_pdim} gives
\[
d \;\le\; 4\Big(2 \cdot 2\,N\,m\,\chi^2\,\log(12N) + 1\Big)\log_2(6),
\qquad m=p+1.
\]
Since $\log_2(6) \approx 2.585$, this simplifies to
\[
d \;\lesssim\; 41.4\,N\,m\,\chi^2\,\log(12N) + 10.3
\;=\; O\big(N\,m\,\chi^2\,\log(N)\big).
\]
\end{remark}

\subsection{Uniform convergence and ERM}

We apply a standard uniform deviation bound for bounded loss classes in terms of pseudo-dimension
\cite[Thm. 11.8]{Mohri2018}.

\begin{theorem}[Uniform deviation for bounded losses via pseudo-dimension {\cite[Thm. 11.8]{Mohri2018}}]\label{thm:pdim_uc}
Let $\cG$ be a class of measurable loss functions $\ell:\mathcal{Z}\to[0,M]$ on $\mathcal{Z}=\mathcal{X}\times\mathcal{Y}$,
and let $d:=\mathrm{Pdim}(\cG)$. Then for any $\delta\in(0,1)$, with probability at least $1-\delta$
over $n$ i.i.d.\ samples $Z_1,\dots,Z_n\sim P$,
\[
\sup_{\ell\in\cG}\left|\E[\ell(Z)]-\frac1n\sum_{i=1}^n \ell(Z_i)\right|
\;\le\;
M\sqrt{\frac{2d\log\!\big(\frac{en}{d}\big)}{n}}
\;+\;
M\sqrt{\frac{\log(1/\delta)}{2n}}.
\]
\end{theorem}

Applying Theorem \ref{thm:pdim_uc} to the squared-loss class
$\cG=\{(x,y)\mapsto (h(x)-y)^2: h\in\cH_{\mathrm{TT}}(\Lambda,\chi)\}$ gives the usual ERM excess-risk bound.


\begin{theorem}[ERM excess-risk bound]\label{thm:pdim_erm}
Assume regression model as defined in Definition \ref{ass:regression} and let $\hat h$ be an ERM over $\cH_{\mathrm{TT}}(\Lambda,\chi)$.
Let $\cG=\{(x,y)\mapsto (h(x)-y)^2: h\in\cH_{\mathrm{TT}}(\Lambda,\chi)\}$ and set
$M=(\Lambda K^N+Y_{\max})^2$. Let $d=\mathrm{Pdim}(\cG)$.
Then with probability at least $1-\delta$:
\[
L(\hat h)\le \inf_{h\in\cH_{\mathrm{TT}}(\Lambda,\chi)} L(h)
\;+\;
2M\sqrt{\frac{2d\log\!\big(\frac{en}{d}\big)}{n}}
\;+\;
2M\sqrt{\frac{\log(1/\delta)}{2n}}.
\]
The same bound holds with $L$ replaced by $R$ (Lemma \ref{lem:cleanrisk}).
\end{theorem}

\begin{proof}
For each $h\in\cH_{\mathrm{TT}}(\Lambda,\chi)$ define $\ell_h(Z):=(h(X)-Y)^2$, so that
$L(h)=\E[\ell_h(Z)]$ and $\hat L_n(h)=\frac1n\sum_{i=1}^n \ell_h(Z_i)$.
By Theorem \ref{thm:pdim_uc}, with probability at least $1-\delta$,
\[
\sup_{h\in\cH_{\mathrm{TT}}(\Lambda,\chi)}\big|L(h)-\hat L_n(h)\big|\le \Delta,
\qquad
\Delta :=
M\sqrt{\frac{2d\log\!\big(\frac{en}{d}\big)}{n}}
+
M\sqrt{\frac{\log(1/\delta)}{2n}}.
\]
Then for any $h\in\cH_{\mathrm{TT}}(\Lambda,\chi)$,
\[
L(\hat h)\le \hat L_n(\hat h)+\Delta \le \hat L_n(h)+\Delta \le L(h)+2\Delta.
\]
Taking the infimum over $h$ yields the claim. The transfer from $L$ to $R$ follows from
Lemma \ref{lem:cleanrisk}.
\end{proof}

\subsection{Consolidated learning guarantee for local TT surrogates (LTTS) (Taylor + TT + estimation)}

\begin{lemma}[Uniform pointwise error implies risk bound]\label{lem:unif_to_risk}
If a measurable $h$ satisfies $|h(x)-g(x)|\le E$ for all $x\in\cB(x_0,r)$,
then $R(h)=\E[(h(X)-g(X))^2]\le E^2$.
\end{lemma}
\begin{proof}
Since $|h(x)-g(x)|\le E$ for all $x\in\cB(x_0,r)$, 
squaring both sides gives $(h(x)-g(x))^2 \le E^2$ for all $x\in\cB(x_0,r)$.
By Definition \ref{ass:regression}, $X$ is supported on $\cB(x_0,r)$.
Therefore, $(h(X)-g(X))^2 \le E^2$ almost surely, and taking expectations yields
\[
R(h) = \E[(h(X)-g(X))^2] \le E^2. \qedhere
\]
\end{proof}

\begin{theorem}[Consolidated PAC bound for local TT surrogates (LTTS)]\label{thm:end_to_end_pdim}
Let $g$ and regression model be defined by Definitions \ref{ass:smooth} and \ref{ass:regression}. Fix $(\Lambda,\chi)$ and let $\hat h$ be an ERM
over $\cH_{\mathrm{TT}}(\Lambda,\chi)$. Define
\[
E_{\mathrm{det}}
:=
\frac{C_{p+1}^{(x)} \, r^{p+1} \, N^{p+1}}{(p+1)!}
\;+\;
K^N\,\varepsilon_{\mathrm{TT}}(\Lambda,\chi).
\]
Let $\cG=\{(x,y)\mapsto (h(x)-y)^2: h\in\cH_{\mathrm{TT}}(\Lambda,\chi)\}$, set $d=\mathrm{Pdim}(\cG)$ and
$M=(\Lambda K^N+Y_{\max})^2$. Then with probability at least $1-\delta$:
\[
R(\hat h)
\;\le\;
E_{\mathrm{det}}^2
\;+\;
2M\sqrt{\frac{2d\log\!\big(\frac{en}{d}\big)}{n}}
\;+\;
2M\sqrt{\frac{\log(1/\delta)}{2n}}.
\]
Moreover, by Lemma \ref{lem:loss_pdim_sq} and Corollary \ref{cor:tt_pdim}, one may take

\[
d \;=\; \mathrm{Pdim}(\cG)
\;\le\; 4(2\,\mathrm{Pdim}(\cH) + 1)\log_2(6)
\]
which is of the order $ O\big(N\,m\,\chi^2\,\log(N)\big)$ where $m=p+1$
\end{theorem}

\begin{proof}
Fix an arbitrary $\tau>0$. By definition of $\varepsilon_{\mathrm{TT}}(\Lambda,\chi)$ (Definition \ref{def:epsTT}),
there exists $A_\tau\in\cA_{\mathrm{TT}}(\Lambda,\chi)$ such that
\[
\|A^\star-A_\tau\|_F \le \varepsilon_{\mathrm{TT}}(\Lambda,\chi)+\tau.
\]
Apply Theorem \ref{thm:deterministic} with $A_{\mathrm{TT}}=A_\tau$ to obtain a uniform pointwise bound on the patch:
\[
|g(x)-h_{A_\tau}(x)|
\le
\frac{C_{p+1}^{(x)} \, r^{p+1} \, N^{p+1}}{(p+1)!}
+
K^N\big(\varepsilon_{\mathrm{TT}}(\Lambda,\chi)+\tau\big)
\qquad \forall x\in\cB(x_0,r).
\]
Hence, by Lemma \ref{lem:unif_to_risk},
\[
R(h_{A_\tau}) \le \left(
\frac{C_{p+1}^{(x)} \, r^{p+1} \, N^{p+1}}{(p+1)!}
+
K^N\big(\varepsilon_{\mathrm{TT}}(\Lambda,\chi)+\tau\big)
\right)^2.
\]
Therefore,
\[
\inf_{h\in\cH_{\mathrm{TT}}(\Lambda,\chi)} R(h) \le R(h_{A_\tau}).
\]
Apply Theorem \ref{thm:pdim_erm} with $R$ in place of $L$ (via Lemma \ref{lem:cleanrisk}) and combine.
Finally let $\tau\downarrow 0$.
\end{proof}

As discussed before, the deterministic quantity $E_{\mathrm{det}}$ in Theorem \ref{thm:end_to_end_pdim} is an
\emph{error certificate} for the \emph{reference predictor} $h_{A_\tau}$ obtained by
(Taylor truncation $\rightarrow$ embedding $\rightarrow$ TT compression), as established by
Theorem \ref{thm:deterministic}. It is not claimed that the same Taylor-remainder term holds
for an arbitrary ERM minimizer $\hat h\in\cH_{\mathrm{TT}}(\Lambda,\chi)$, which may in principle
place mass on indices with $|\alpha|>p$.
Instead, we use $h_{A_\tau}$ only to upper bound the approximation term
$\inf_{h\in\cH_{\mathrm{TT}}(\Lambda,\chi)}R(h)\le R(h_{A_\tau})$, and then apply the ERM
uniform-convergence bound to conclude that $\hat h$ competes with (and may improve upon)
this reference predictor.

\begin{corollary}[Sample size for excess-risk tolerance $\eta$]\label{cor:sample_pdim}
Fix $\eta>0$ and $\delta\in(0,1)$. If
\[
n \;\gtrsim\; \frac{(\Lambda K^N+Y_{\max})^4}{\eta^2}\left(d\log\!\Big(\frac{en}{d}\Big)+\log\frac{2}{\delta}\right),
\]
then with probability at least $1-\delta$ the ERM satisfies
\[
R(\hat h)\le \inf_{h\in\cH_{\mathrm{TT}}(\Lambda,\chi)} R(h)+\eta.
\]
\end{corollary}

\begin{proof}
Theorem \ref{thm:end_to_end_pdim} yields a risk bound $R(\hat h)
\;\le\;
E_{\mathrm{det}}^2
\;+\; \Delta_n(\delta)$. To obtain a sample complexity statement, we fix a target estimation tolerance $\eta\;>\;0$ and require $\Delta_n(\delta)\;\le\;\eta$. Solving this sufficient condition for $n$ and substituting for $M=(\Lambda K^N+Y_{\max})^2$ yields the above sample complexity.  
\end{proof}


\subsection{Local TT surrogates (LTTS) guarantee: risk bound and sample complexity}

The following theorem is the main statistical result of this paper. It 
specializes the general statistical learning machinery of the preceding sections to the 
local hypothesis class $\cH_{\mathrm{TT}}(\Lambda^*(r),\chi)$, making the full 
$r$-dependence explicit and quantifying the advantage of local over 
global surrogation.

\begin{theorem}[LTTS guarantee with $r$ dependence]\label{thm:local_ltts}
Let $g$ and the regression model be as in Definitions \ref{ass:smooth} 
and \ref{ass:regression}. Let $\hat{h}$ be an ERM over the local hypothesis 
class $\cH_{\mathrm{TT}}(\Lambda^*(r),\chi)$ from Definition \ref{def:localhypclass}, 
with $\Lambda^\star(r)$ from Lemma \ref{lem:lambda_star}. Define
\[
M(r) \;:=\; \bigl(\Lambda^\star(r)\,K^N + Y_{\max}\bigr)^2,
\qquad
\varepsilon_{\mathrm{TT}}(r,\chi) \;:=\; \varepsilon_{\mathrm{TT}}
\bigl(\Lambda^\star(r),\chi\bigr),
\]
where $\varepsilon_{\mathrm{TT}}(\Lambda^\star(r),\chi)$ is as in 
Definition \ref{def:epsTT}, and let $d = \mathrm{Pdim}(\cG) = 
O(N\,m\,\chi^2\log N)$ with $m = p+1$.

\textbf{(i) Risk bound.}
With probability at least $1-\delta$ over $n$ i.i.d.\ samples:
\[
R(\hat{h})
\;\le\;
\underbrace{
\left(
\frac{C_{p+1}^{(x)}\,r^{p+1}\,N^{p+1}}{(p+1)!}
\;+\;
K^N\,\varepsilon_{\mathrm{TT}}(r,\chi)
\right)^{\!2}
}_{E_{\mathrm{det}}(r)^2\ \text{(approximation)}}
\;+\;
\underbrace{
2M(r)\sqrt{\frac{2d\log\!\bigl(\frac{en}{d}\bigr)}{n}}
\;+\;
2M(r)\sqrt{\frac{\log(1/\delta)}{2n}}
}_{\text{statistical terms}}.
\]

\textbf{(ii) Sample complexity.}
Fix excess-risk tolerance $\eta > 0$ and $\delta \in (0,1)$. 
To guarantee $R(\hat{h}) \le E_{\mathrm{det}}(r)^2 + \eta$ with 
probability at least $1-\delta$, it suffices to take
\[
n \;\gtrsim\;
\frac{M(r)^2}{\eta^2}
\left(
d\log\!\Big(\frac{en}{d}\Big) + \log\frac{2}{\delta}
\right).
\]

\end{theorem}

\begin{proof}
Apply Theorem \ref{thm:end_to_end_pdim} with $\Lambda = \Lambda^\star(r)$ 
from Lemma \ref{lem:lambda_star}, also
$\cA_{\mathrm{TT}}(\Lambda^\star(r), \chi)$ is non-empty and hence 
$\varepsilon_{\mathrm{TT}}(r,\chi) < \infty$. 
Substituting $M = M(r)$ into the PAC bound of 
Theorem \ref{thm:end_to_end_pdim} gives part (i) directly.
Part (ii) follows from Corollary \ref{cor:sample_pdim} under the same 
substitution.
\end{proof}

Both the Taylor truncation term in $E_{\mathrm{det}}(r)^2$ and $M(r)$ decrease and saturate as $r\to 0$, so 
smaller patches simultaneously reduce approximation error and statistical 
complexity. Note that $M(r)$ contains a factor of $K^N$; for fixed $N$ 
this is a constant, but for large $N$ the choice of $r$ and $\chi$ must 
be made carefully to control the overall bound. Smaller patches 
require fewer samples to achieve the same excess-risk tolerance $\eta$, 
giving a quantitative advantage of local over global surrogation.

\section{Application to Quantum Learning Models}\label{sec:quantum}

The local TT surrogate (LTTS) framework developed in the preceding sections applies to any locally smooth
black-box function with bounded signal and noise. A gate-based quantum learning model typically outputs an expectation value
\[
  g(x) = \Tr\!\bigl(O\, U(x)\,\rho\, U(x)^\dagger\bigr),
\]
where $x \in \R^N$ is a classical input, $\rho$ is an initial quantum state,
$U(x)$ is a data-dependent unitary circuit, and $O$ is a bounded Hermitian
observable. When the input enters through finitely many gates of the form
$e^{-i x_j H_{j,\ell}}$ with finite-dimensional Hermitian generators
$H_{j,\ell}$, the map $x \mapsto U(x)$ is real-analytic, since the matrix
exponential is an entire function of the scalar parameter and products of
analytic maps are analytic. Because the trace is linear, $g(x)$ is
real-analytic in $x$. In particular, $g \in C^{p+1}(\cB(x_0,r))$ for any $p$
and any compact patch $\cB(x_0,r)$. Since each partial derivative
$\partial_x^\alpha g$ is continuous and $\cB(x_0,r)$ is compact, the supremum
in Definition \ref{ass:smooth} is finite. This reasoning applies broadly: it
covers reuploading models \cite{P_rez_Salinas_2020,PhysRevA.103.032430}, quantum convolutional architectures (as demonstrated in experiment \ref{sub:ltts_qcnn}) and general hybrid
ans\"atze with data-dependent rotations.

For bounded observables, $|g(x)| \le \|O\|_\infty$, hence we may take
$G_{\max}=\|O\|_\infty$. Moreover, under finite-shot estimation of a bounded
observable, each label $Y_i$ (a sample mean, hence a convex combination of
bounded outcomes) satisfies $|Y_i|\le \|O\|_\infty$ deterministically.
Therefore the induced noise $\varepsilon_i:=Y_i-g(X_i)$ is also bounded, with
$|\varepsilon_i|\le |Y_i|+|g(X_i)| \le 2\|O\|_\infty$, so
Definition \ref{ass:regression} holds with $\sigma = 2\|O\|_\infty$ and hence
$Y_{\max}=G_{\max}+\sigma \le 3\|O\|_\infty$. In our setting the tighter label
bound $|Y_i|\le \|O\|_\infty$ is available directly and can be used wherever
only a bound on $|Y_i|$ is required. Such local smoothness and boundedness are precisely the conditions our statistical analysis requires. Thus, trained quantum models are a natural setting for this framework. The high cost of circuit evaluation makes the speedup from classical surrogation especially valuable.


\section{Numerical Experiments}
\subsection{Rank scaling experiment}

\paragraph{Oracle experiment: TT-rank scaling of simplex-to-box zero-padding.}
We test the effects of embedding a simplex-supported derivative tensor into a box-indexed tensor through zero-padding on the TT rank. Given a test function $f\colon\mathbb{R}^N\to\mathbb{R}$ (defined over tensor dimension $N$) and truncation
order $p$, we construct two \emph{derivative tensors} over the index set
$\alpha\in\{0,\dots,p\}^N$:
\begin{itemize}
\item the \emph{true box tensor} on the box index set $\alpha\in\{0,\dots,p\}^N$,
namely
\[
A_{\mathrm{box}}[\alpha] \;=\; \partial^\alpha f(x_0),
\]
\item and the \emph{simplex-embedded tensor} $A_\Delta$ defined on the same box shape by
\[
A_\Delta[\alpha] \;=\;
\begin{cases}
\partial^\alpha f(x_0), & |\alpha|\le p,\\
0, & |\alpha|>p.
\end{cases}
\]
\end{itemize}

Each tensor is approximated by TT-SVD at successive rank caps
$\chi\in\{1,\dots,\chi_{\max}\}$ via left-to-right sweeps, and we record the smallest $\chi$ for
which the reconstruction $A_{\mathrm{TT}}$ meets a preset Frobenius tolerance.
Let $\chi_{\mathrm{box}}(\varepsilon)$ and $\chi_{\Delta}(\varepsilon)$
denote the corresponding smallest rank caps for $A_{\mathrm{box}}$ and
$A_\Delta$ respectively, and define the rank ratio $\rho=\chi_\Delta/\chi_{\mathrm{box}}$.
We report $\rho$ under two normalisations: a \emph{common-scale} criterion:
\begin{equation}\label{eq:common-scale}
  \lVert A - A_{\mathrm{TT}}\rVert_F
  \;\le\;
  \varepsilon\,\lVert A_{\mathrm{box}}\rVert_F,
\end{equation}
applied identically to both tensors, and a \emph{self-relative} criterion:
\begin{equation}\label{eq:self-relative}
  \lVert A - A_{\mathrm{TT}}\rVert_F
  \;\le\;
  \varepsilon\,\lVert A\rVert_F,
\end{equation}
applied to each tensor with its own norm.
Figure \ref{fig:rank-scatter} visualises all
$(\chi_{\mathrm{box}},\,\chi_\Delta)$ pairs under both criteria at
$\varepsilon\in\{10^{-2},10^{-3}\}$, and Table \ref{tab:rank-summary}
reports medians, interquartile ranges, and median rank caps per function
family.

Under the experiment, we tested seven function classes (\emph{ExpSum, ProductCos, matched-degree polynomials, higher-degree polynomials, QuadraticForm, Trig and Gauss}) spanning a range of separability and
smoothness structures across configurations
$(N,p)\in\{(4,3),(4,4),(6,2),(6,3)\}$ and tolerances
$\varepsilon\in\{10^{-2},10^{-3}\}$, yielding 176 test cases per
tolerance level in our default configuration.
Table \ref{tab:rank-summary} aggregates these into five
buckets (Separable, Poly $\deg=p$, Poly $\deg>p$, QuadraticForm, Trig\,+\,Gauss). All cases reach the prescribed tolerance within the rank cap
$\chi_{\max}=25$ (0 unreached across all families and tolerances),
ruling out selection effects.


\begin{figure}[!t]
  \centering
  \includegraphics[width=\textwidth]{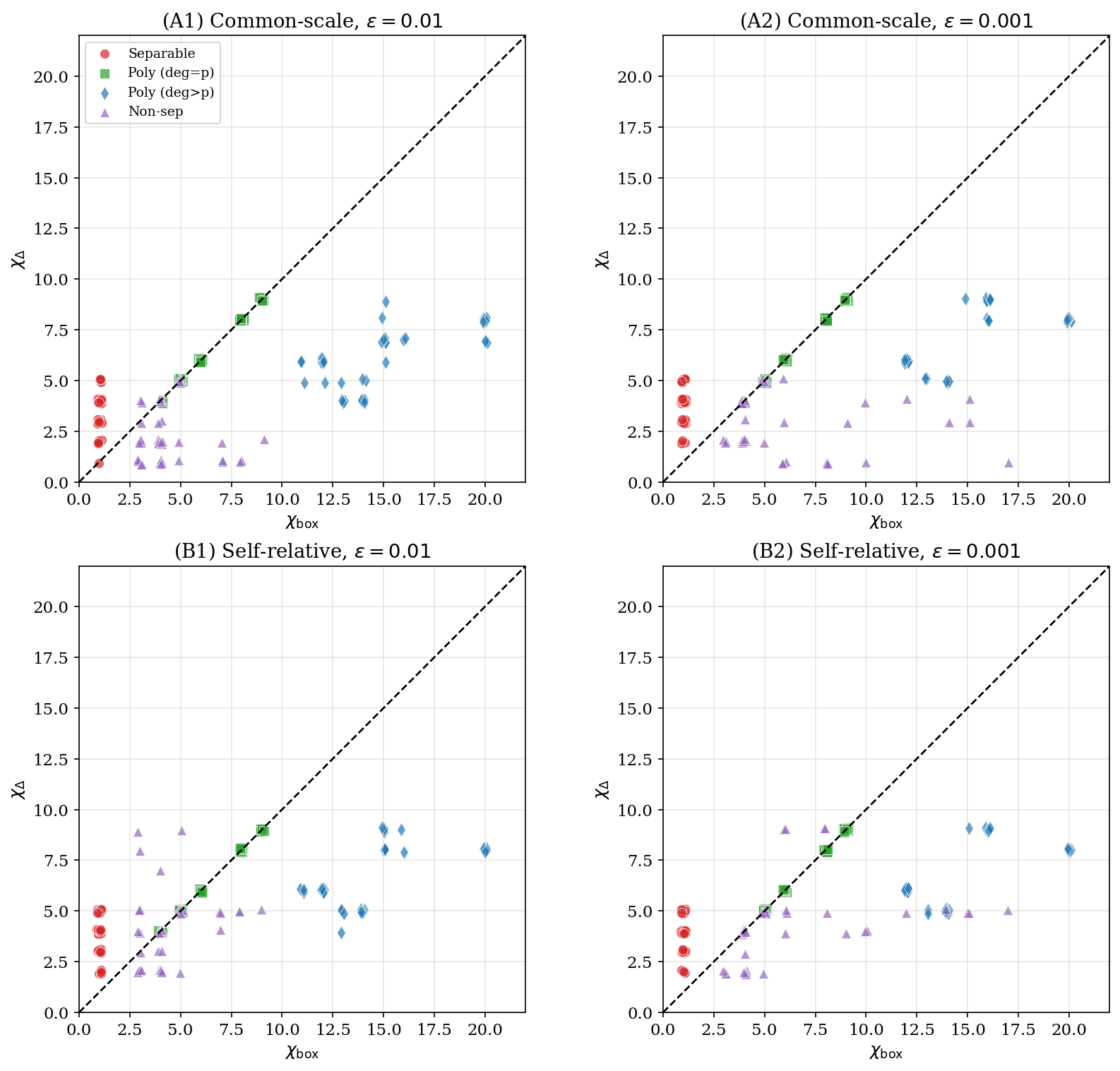}
  \caption{%
    Scatter of minimal TT rank caps
    $(\chi_{\mathrm{box}},\,\chi_\Delta)$ for 176 test cases across
    four $(N,p)$ configurations.
    Top row: common-scale criterion
    $\lVert A-A_{\mathrm{TT}}\rVert_F \le
     \varepsilon\lVert A_{\mathrm{box}}\rVert_F$;
    bottom row: self-relative criterion
    $\lVert A-A_{\mathrm{TT}}\rVert_F \le
     \varepsilon\lVert A\rVert_F$.
    Left column: $\varepsilon=10^{-2}$;
    right column: $\varepsilon=10^{-3}$.
    Points above the dashed diagonal indicate rank inflation
    ($\rho>1$); points below indicate deflation ($\rho<1$).
    Small random jitter is added for visibility since ranks are
    discrete.
    Separable functions cluster above the diagonal at
    $\chi_{\mathrm{box}}=1$; matched-degree polynomials lie on the
    diagonal; higher-degree polynomials and additional non-separable
    families fall below the diagonal under the common-scale criterion
    and move closer to the diagonal under the self-relative criterion.}
  \label{fig:rank-scatter}
\end{figure}


\begin{table}[t]
  \centering
  \caption{%
    Median rank ratio $\rho = \chi_\Delta / \chi_{\mathrm{box}}$
    with interquartile range (IQR) and median rank caps per
    function family, under both tolerance normalisations.
    All 176 instances reach the prescribed tolerance within
    $\chi_{\max}=25$ (0 unreached).}
  \label{tab:rank-summary}
  \small
  \setlength{\tabcolsep}{4.5pt}
  \begin{tabular}{@{}ll r cc cc@{}}
    \toprule
    & & &
    \multicolumn{2}{c}{\textbf{Common-scale}} &
    \multicolumn{2}{c}{\textbf{Self-relative}} \\
    \cmidrule(lr){4-5} \cmidrule(lr){6-7}
    \textbf{Family} & $\varepsilon$ & $n$ &
      med.\,$\rho$ [IQR] & med.\,$\chi_{\mathrm{box}}$/$\chi_\Delta$ &
      med.\,$\rho$ [IQR] & med.\,$\chi_{\mathrm{box}}$/$\chi_\Delta$ \\
    \midrule
    \multirow{2}{*}{Separable}
      & $10^{-2}$ & 48
        & 4.00\;[3.00,\,4.00] & 1\,/\,4
        & 4.00\;[3.00,\,4.00] & 1\,/\,4 \\
      & $10^{-3}$ & 48
        & 4.00\;[3.00,\,4.00] & 1\,/\,4
        & 4.00\;[3.00,\,4.00] & 1\,/\,4 \\
    \addlinespace
    \multirow{2}{*}{Poly ($\deg\!=\!p$)}
      & $10^{-2}$ & 40
        & 1.00\;[1.00,\,1.00] & 7\,/\,7
        & 1.00\;[1.00,\,1.00] & 7\,/\,7 \\
      & $10^{-3}$ & 40
        & 1.00\;[1.00,\,1.00] & 7\,/\,7
        & 1.00\;[1.00,\,1.00] & 7\,/\,7 \\
    \addlinespace
    \multirow{2}{*}{Poly ($\deg\!>\!p$)}
      & $10^{-2}$ & 40
        & 0.40\;[0.36,\,0.47] & 14\,/\,6
        & 0.45\;[0.40,\,0.53] & 14\,/\,7 \\
      & $10^{-3}$ & 40
        & 0.45\;[0.40,\,0.50] & 14\,/\,7
        & 0.45\;[0.40,\,0.52] & 14\,/\,7 \\
    \addlinespace
    \multirow{2}{*}{QuadraticForm}
      & $10^{-2}$ & 16
        & 1.00\;[1.00,\,1.00] & 4\,/\,4
        & 1.00\;[1.00,\,1.00] & 4\,/\,4 \\
      & $10^{-3}$ & 16
        & 1.00\;[1.00,\,1.00] & 4\,/\,4
        & 1.00\;[1.00,\,1.00] & 4\,/\,4 \\
    \addlinespace
    \multirow{2}{*}{Trig\,+\,Gauss}
      & $10^{-2}$ & 32
        & 0.37\;[0.25,\,0.67] & 4\,/\,2
        & 0.69\;[0.57,\,1.33] & 4\,/\,4 \\
      & $10^{-3}$ & 32
        & 0.50\;[0.19,\,0.54] & 6\,/\,2
        & 0.50\;[0.44,\,0.83] & 6\,/\,4 \\
    \bottomrule
  \end{tabular}
\end{table}


\paragraph{Numerical settings:}
Derivatives are evaluated at $x_0=0$ and tensors are formed on the dense
grid $\{0,\dots,p\}^N$ of size $(p+1)^N$.
For separable baselines we used \emph{ExpSum}
$f(x)=\exp\!\bigl(\sum_i c_i x_i\bigr)$ and \emph{ProductCos}
$f(x)=\prod_i\cos(c_i x_i)$, each with a deterministic all-ones instance
and 10 random draws ($c_i\sim\mathcal{N}(0,1)$) per $(N,p)$.
For \emph{polynomials} we tested \emph{matched-degree} ($\deg=p$) and \emph{higher-degree}
($\deg=p+2$), each with 10 independent random seeds per degree class per
$(N,p)$ to safeguard against single-instance artifacts.
Additional non-separable families comprised:
(i) coupled quadratic forms (\emph{QuadraticForm})
    $f(x)=x^{\!\top}\! Ax + b^{\!\top}\! x + c$ with dense $A$
    (4 random instances per $(N,p)$);
(ii) trigonometric sums (\emph{Trig})
    $f(x)=\sum_{k=1}^{K} w_k \prod_{i=1}^{N}\cos(m_{k,i} x_i)$ with
    $K=8$ and $\lVert m_k\rVert_1 \le 3$
    (4 random instances per $(N,p)$); and
(iii) correlated Gaussians (\emph{Gauss})
    $f(x)=\exp\!\bigl(-\tfrac{1}{2}\,x^{\!\top}\! Qx\bigr)$ with random
    SPD precision $Q$
    (4 random instances per $(N,p)$).
For the Gaussian family, derivatives were computed exactly at $x_0=0$ via a
pairing recursion (Wick/Isserlis); this expansion point causes all
odd-total-order derivatives to vanish, introducing structured sparsity in both tensors.

The results split into three regimes, visible as distinct clusters in
Figure \ref{fig:rank-scatter} and confirmed quantitatively in
Table \ref{tab:rank-summary}.
\begin{itemize}
\item \emph{Separable functions} (ExpSum, ProductCos) exhibit strong, consistent
rank inflation: the corresponding points sit well above the diagonal at
$\chi_{\mathrm{box}}=1$, with median $\rho=4.00$ and IQR
$[3.00,\,4.00]$ across both tolerances and both normalisations
(Table \ref{tab:rank-summary}).
These functions admit a TT-rank-1 representation on the
box, and the simplex mask disrupts this low-rank separable structure by
introducing cross-mode coupling.

\item \emph{Matched-degree polynomials} ($\deg=p$) yield $\rho=1.00$ in every
instance i.e for all 40 draws, every tolerance, both normalisations, appearing
exactly on the diagonal in Figure \ref{fig:rank-scatter}. The quadratic forms, whose Taylor expansions terminate at order 2 and thus
lie entirely within the simplex for all tested $p\ge 2$, show $\rho=1.00$
throughout.

\item \emph{Higher-degree polynomials} ($\deg>p$) and several additional
\emph{non-separable families} show rank deflation under the common-scale
criterion.
In the top row of Figure \ref{fig:rank-scatter}, Poly ($\deg>p$) and the
Trig\,+\,Gauss class fell consistently below the diagonal.
Table \ref{tab:rank-summary} reports that at
$\varepsilon=10^{-2}$, Poly ($\deg>p$) has median $\rho=0.40$ and
Trig\,+\,Gauss has median $\rho=0.37$ under common-scale normalisation.
This deflation is consistent with simplex truncation removing
high-total-degree content that would otherwise require approximation. Under the self-relative criterion (bottom row), the Trig\,+\,Gauss class
moved closer to the diagonal relative to the common-scale panels, and its
interquartile range widens to include values above unity at
$\varepsilon=10^{-2}$ (Table \ref{tab:rank-summary}: median $\rho=0.69$,
IQR $[0.57,\,1.33]$), indicating that zero-padding can be mildly inflationary
for some instances when measured relative to each tensor's own norm.

\end{itemize}

In summary, although Trig\,+\,Gauss class under self-relative criterion showed mild rank inflation, the median behaviour for several non-separable families
remains deflationary, showing
that the direction and magnitude of rank change depend on both tensor
content and the chosen normalisation. Across the tested non-separable families, simplex-to-box embedding via zero-padding does not exhibit systematic inflation in the tested non-separable families; under common-scale it is predominantly deflationary, while under self-relative it is closer to neutral with occasional mild inflation. This validates that our implementation choice (box-index TT storage with simplex truncation via zero-padding) is not inherently penalising in rank for the tested non-separable families.

\subsection{Local TT surrogate (LTTS) empirical validation} \label{sub:ltts_qcnn}


\begin{figure}[t]
    \centering
    \includegraphics[width=\linewidth]{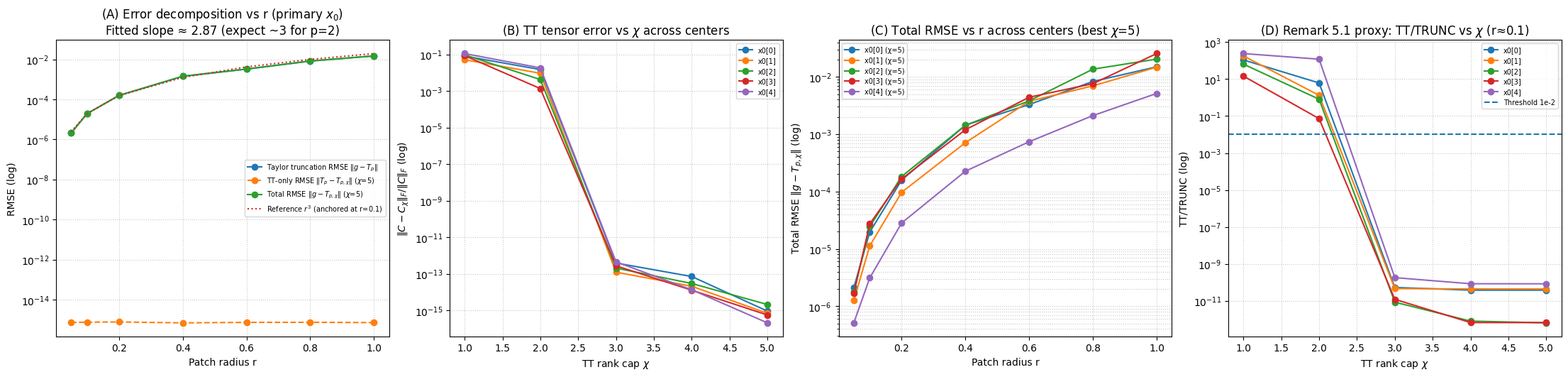}\par\vspace{0.6em}
    \includegraphics[width=\linewidth]{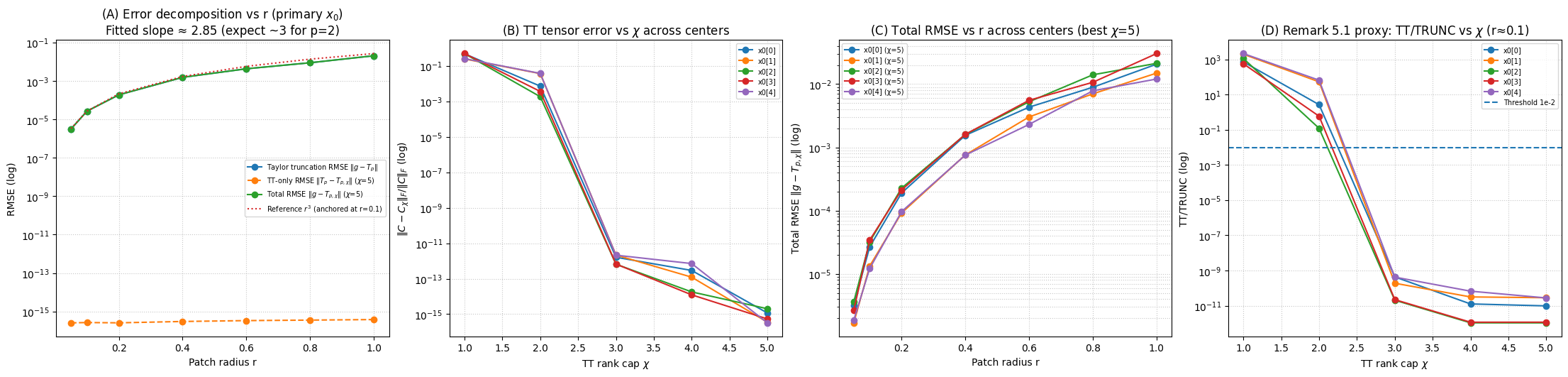}
    \caption{\textbf{Local TT surrogate (LTTS) validation across two datasets:}
    \textbf{Top row:} synthetic Gaussian classification. \textbf{Bottom row:} UCI Banknote Authentication.
    Each row shows the same 4-panel diagnostic: (A) error decomposition vs patch radius $r$ at the primary center $x_0$,
    including an $r^3$ reference,
    (B) TT coefficient-tensor truncation error $\|C-C_\chi\|_F/\|C\|_F$ vs rank cap $\chi$ across multiple centers;
    (C) total surrogate RMSE $\|g-T_{p,\chi}\|$ vs $r$ across centers at the best (largest) $\chi$ used;
    (D) Remark \ref{rem:informative} proxy, TT/TRUNC $=\mathrm{RMSE}(T_p-T_{p,\chi})/\mathrm{RMSE}(g-T_p)$ vs $\chi$ at $r\approx 0.1$,
    with the $10^{-2}$ threshold indicating when TT compression error is negligible relative to Taylor truncation.}
    \label{fig:local-tt-validation-two-datasets}
\end{figure}


We validate and demonstrate the core claims behind the local TT surrogate (LTTS) construction using experiments on two independent datasets:
(i) a synthetic Gaussian dataset  and (ii) a real-data experiment on the UCI Banknote Authentication dataset. In both cases, we train an identical 6-qubit QCNN ($N\!=\!20$ parameters, angle embedding $R_Y(x_i)$, two entangling layers, pooling to
qubit 0, output $g(x) = \langle Z_0 \rangle$). After training and freezing the parameters $\theta^*$, the function $g$ is treated as a black-box oracle. Then we construct a degree-$p=2$ local Taylor surrogate around multiple centers $x_0$ and compress its coefficient tensor using TT-SVD with rank cap $\chi$.
We sweep (a) patch centers $x_0$ and (b) patch radii $r$, and report a consistent set of diagnostics summarized in Fig. \ref{fig:local-tt-validation-two-datasets}.

\paragraph{Datasets:}
\textbf{(1) Synthetic Gaussian} ($D\!=\!6$): balanced binary classification with class
means separated by $\Delta\!=\!2$ in coordinate 1; features clipped and mapped to
$[0,\pi]$. Trained with 3 restarts to 96\% test accuracy.
\textbf{(2) UCI Banknote Authentication \cite{banknote_authentication_267}} ($D\!=\!6$): 1{,}372 samples with 4 wavelet
features, standardised, zero-padded to 6D, mapped to $[0,\pi]$. Trained with 1 restart
to 86\% test accuracy.

\paragraph{Pipeline:}
At each chosen patch centre $x_0$, we: (i) compute the degree-$p\!=\!2$ Taylor coefficient
tensor $\mathcal{C}$ via 73 central finite-difference circuit
evaluations, (ii) decompose
$\mathcal{C}$ into TT format via TT-SVD with hard rank cap $\chi$,
(iii) sweep $\chi \in \{1,\dots,5\}$ to identify the threshold where TT compression
error becomes negligible (this is done to back our remark \ref{rem:informative}); (iv) sweep patch radius
$r \in \{0.05,\dots,1.0\}$ to verify $O(r^3)$ truncation scaling;
(v) fit an ERM surrogate via ALS warm-started with the Taylor-TT surrogate on 600 uniform samples in $B_r(x_0)$ and
(vi) verify the deterministic certificate bound $E_{\det}$.
To test spatial universality, we repeat at 5 diverse centres per dataset: the original
training sample, the last training sample, a class-0 and class-1 sample near the
decision boundary, and a high-confidence test sample.

All findings are consistent across both datasets and all chosen patch centres. In Fig. \ref{fig:local-tt-validation-two-datasets} for both datasets, the four panels separate and validate the two error sources in the local Taylor-TT surrogate. (A) shows the primary-patch error decomposition versus $r$ into truncation $\|g-T_p\|$, TT-only $\|T_p-T_{p,\chi}\|$, and total $\|g-T_{p,\chi}\|$, exhibiting the expected $\sim r^3$ scaling for $p=2$ while the TT-only term drops to near machine precision once $\chi$ is sufficiently large, so the total error is truncation-dominated in the shown regime. (B) plots the TT coefficient-tensor compression error $\|C-C_\chi\|_F/\|C\|_F$ versus $\chi$ across several $x_0$, demonstrating rapid decay at modest ranks and consistency across centers. (C) reports the total surrogate RMSE $\|g-T_{p,\chi}\|$ versus $r$ across centers at the best $\chi$, confirming robustness of local accuracy and the expected degradation as patch size grows. And (D) reports the Remark \ref{rem:informative} proxy (validation) $\mathrm{TT/TRUNC}:=\mathrm{RMSE}(T_p-T_{p,\chi})/\mathrm{RMSE}(g-T_p)$ at $r\approx 0.1$, which falls sharply below a small threshold (e.g.\ $10^{-2}$) beyond a modest $\chi$ and stays extremely small thereafter, indicating that TT compression error is negligible relative to Taylor truncation at that patch scale. Higher relative errors at boundary points ($|g(x_0)| \approx 0$) were observed which can be attributed to small
denominators rather than degraded absolute accuracy. The framework's structural properties are
robust to data distribution, feature space, and expansion centre. In addition, we observe a consistent wall-clock speedup of $\sim 250$–$400\times$ per model evaluation when replacing circuit calls by the TT surrogate (across centers and both datasets).

\paragraph{ERM is consistently competitive with, and often improves upon, the certificate in our experiments.}
The ALS-fitted ERM surrogate achieves \emph{lower} test error than the TT-SVD
certificate at every $\chi$ on both datasets during the experiments as reported in Table \ref{tab:erm_vs_certificate_two_datasets}. The certificate inherits the full Taylor truncation error since it is
constructed from $\mathcal{C}$ without seeing $g$ directly. The ERM, fitting on actual
circuit evaluations, partially corrects for the degree-2 remainder. Warm-starting ALS
from certificate cores is advantageous for convergence as compared to random initialization. But the scaling of advantage due to warm-start for larger datasets is yet to be tested. The deterministic bound $E_{\det}$ holds in all cases (\texttt{cert}$\,\le\,E_{\det}$).

\begin{table}[t]
\centering
\caption{ERM improves upon the Taylor-TT certificate across TT rank caps $\chi$ on both the Synthetic Gaussian and Banknote datasets (data-space certification at $r=0.1$, $p=2$, $D=6$). We report here for a primary patch centre case, the deterministic bound $E_{\det}$, certificate and ERM test RMSE, and the ratio $\text{te\_ratio}=\mathrm{ERM\_RMSE}/\mathrm{cert\_RMSE}$. In all cases, $\text{te\_ratio}<1$ and the certificate satisfies $\text{cert}\le E_{\det}$.}
\label{tab:erm_vs_certificate_two_datasets}
\begin{tabular}{c|cccc|cccc}
\toprule
& \multicolumn{4}{c|}{Synthetic Gaussian} & \multicolumn{4}{c}{Banknote Authentication} \\
\cmidrule(lr){2-5}\cmidrule(lr){6-9}
$\chi$
& $E_{\det}$ & cert\_RMSE & ERM\_RMSE & te\_ratio
& $E_{\det}$ & cert\_RMSE & ERM\_RMSE & te\_ratio \\
\midrule
1
& $1.255\mathrm{e}{-}02$ & $1.133\mathrm{e}{-}04$ & $1.056\mathrm{e}{-}04$ & $0.93$
& $1.016\mathrm{e}{-}01$ & $9.484\mathrm{e}{-}04$ & $8.645\mathrm{e}{-}04$ & $0.91$ \\
2
& $3.572\mathrm{e}{-}04$ & $1.703\mathrm{e}{-}05$ & $8.667\mathrm{e}{-}06$ & $0.51$
& $2.389\mathrm{e}{-}04$ & $2.152\mathrm{e}{-}05$ & $9.400\mathrm{e}{-}06$ & $0.44$ \\
3
& $7.895\mathrm{e}{-}05$ & $1.729\mathrm{e}{-}05$ & $8.557\mathrm{e}{-}06$ & $0.49$
& $8.510\mathrm{e}{-}05$ & $2.188\mathrm{e}{-}05$ & $9.110\mathrm{e}{-}06$ & $0.42$ \\
4
& $7.895\mathrm{e}{-}05$ & $1.729\mathrm{e}{-}05$ & $8.621\mathrm{e}{-}06$ & $0.50$
& $8.510\mathrm{e}{-}05$ & $2.188\mathrm{e}{-}05$ & $8.985\mathrm{e}{-}06$ & $0.41$ \\
5
& $7.895\mathrm{e}{-}05$ & $1.729\mathrm{e}{-}05$ & $8.949\mathrm{e}{-}06$ & $0.52$
& $8.510\mathrm{e}{-}05$ & $2.188\mathrm{e}{-}05$ & $9.124\mathrm{e}{-}06$ & $0.42$ \\
\bottomrule
\end{tabular}
\end{table}

\section{Discussion}

In summary, we proposed a model-agnostic analytical framework for local TT surrogates (LTTS) wherein the total approximation error decomposes into three transparent sources, each independently controllable:
\begin{itemize}
    \item \emph{Polynomial truncation error}, controlled by the patch radius $r$ and polynomial degree $p$;
    \item \emph{Tensor compression error}, controlled by the TT bond dimension $\chi$;
    \item \emph{Statistical estimation error}, arising when the surrogate is learned from finite noisy samples with $r$ dependence.
\end{itemize}

Neither the deterministic approximation component nor the statistical component alone suffices: the Taylor-TT certificate without ERM gives no learning algorithm, and ERM without the certificate gives no guarantee that the hypothesis class is rich enough. The final learned surrogate is a TT-format function but is \emph{not} the certificate predictor. ERM is free to find any predictor in the hypothesis class, and may outperform the certificate. The resulting surrogate can be evaluated, differentiated, and integrated at negligible classical cost. Our framework directly addresses a practical concern in the deployment of quantum machine learning: large-scale-trained quantum models still require repeated quantum circuit executions, which are computationally expensive or even infeasible. By decoupling the training phase (which runs on quantum hardware) from the deployment/inference phase (which runs classically), local TT surrogates enable a workflow in which quantum resources are invested once during training and then amortized over an unlimited number of cheap classical queries. Given the explicit $r$-dependence in our guarantees, local surrogates do improve with smaller patch radii. Numerical experiments on quantum convolutional neural networks trained on real-world datasets validate the theoretical predictions while using the Taylor-TT certificate as a warm start.

A natural question is whether one can obtain guarantees similar to classical surrogates in Schriber et al \cite{Schreiber_2023} by sampling $g$ on a deterministic grid and reconstructing uniformly on $\cB(x_0, r)$. Such results require the target to belong \emph{a priori} to a finite-dimensional function family with a well-conditioned sampling operator (in their case: bandlimited Fourier models). In essence, there ought to be structural guarantees for the quantum learning model we intend to surrogate. The statistical regression approach sidesteps these issues by foregoing any such design prerequisites and views the target as a black box with local smoothness. For such agnostic targets, the TT constraint controls model complexity, Taylor features control the approximation-to-function link, and i.i.d. sampling yields principled PAC guarantees.

In our framework, while the \emph{parameter count} $d_{\mathrm{eff}}$ scales polynomially i.e $\deff = N(p+1)\chi^2$ rather than the naive $(p+1)^N$, the \emph{worst-case constants} in the bounds inherit an exponential factor $K^N$ through the tensor-product feature norm. As a consequence, the worst-case sample complexity bound scales as $K^{4N}$ in Corollory \ref{cor:sample_pdim}, which could be improved upon in future work by utilizing more efficient TT constructions.  This separates \emph{representation complexity} (polynomial via TT structure) from \emph{feature-induced constants} (exponential via the embedding), clarifying precisely where the curse of dimensionality enters. On the same note, it should be possible to construct the certificate predictor using other higher-order approximations such as Chebyshev \cite{rodrguezaldavero2025chebyshevapproximationcompositionfunctions} or Fourier bases, potentially yielding tighter complexity bounds.


\section*{Data availability}
\addcontentsline{toc}{section}{Data availability}  
The code for all the numerical experiments and the plots in this paper is available at \url{https://github.com/sreerajrajindrannair/Local_tensor-train_surrogates_for_qml}

\section*{Acknowledgments}

\addcontentsline{toc}{section}{Acknowledgments}  

S.R.N. acknowledges and thanks his doctoral research funding agency, Sydney Quantum Academy (SQA). The authors thank Afrad Basheer for reviewing an early draft of this manuscript.


\bibliographystyle{unsrturl}
\bibliography{references}

@article{Schreiber_2023,
   title={Classical Surrogates for Quantum Learning Models},
   volume={131},
   ISSN={1079-7114},
   url={http://dx.doi.org/10.1103/PhysRevLett.131.100803},
   DOI={10.1103/physrevlett.131.100803},
   number={10},
   journal={Physical Review Letters},
   publisher={American Physical Society (APS)},
   author={Schreiber, Franz J. and Eisert, Jens and Meyer, Johannes Jakob},
   year={2023},
   month=sep }

@misc{lerch2024efficient,
    title={Efficient quantum-enhanced classical simulation for patches of quantum landscapes},
    author={Sacha Lerch and Ricard Puig and Manuel S. Rudolph and Armando Angrisani and Tyson Jones and M. Cerezo and Supanut Thanasilp and Zoë Holmes},
    year={2024},
    eprint={2411.19896},
    archivePrefix={arXiv},
    primaryClass={quant-ph}
}

@article{PhysRevA.107.062612,
  title = {Multidimensional Fourier series with quantum circuits},
  author = {Casas, Berta and Cervera-Lierta, Alba},
  journal = {Phys. Rev. A},
  volume = {107},
  issue = {6},
  pages = {062612},
  numpages = {15},
  year = {2023},
  month = {Jun},
  publisher = {American Physical Society},
  doi = {10.1103/PhysRevA.107.062612},
  url = {https://link.aps.org/doi/10.1103/PhysRevA.107.062612}
}

@article{PhysRevA.103.032430,
  title = {Effect of data encoding on the expressive power of variational quantum-machine-learning models},
  author = {Schuld, Maria and Sweke, Ryan and Meyer, Johannes Jakob},
  journal = {Phys. Rev. A},
  volume = {103},
  issue = {3},
  pages = {032430},
  numpages = {12},
  year = {2021},
  month = {Mar},
  publisher = {American Physical Society},
  doi = {10.1103/PhysRevA.103.032430},
  url = {https://link.aps.org/doi/10.1103/PhysRevA.103.032430}
}

@article{P_rez_Salinas_2020,
   title={Data re-uploading for a universal quantum classifier},
   volume={4},
   ISSN={2521-327X},
   url={http://dx.doi.org/10.22331/q-2020-02-06-226},
   DOI={10.22331/q-2020-02-06-226},
   journal={Quantum},
   publisher={Verein zur Forderung des Open Access Publizierens in den Quantenwissenschaften},
   author={Pérez-Salinas, Adrián and Cervera-Lierta, Alba and Gil-Fuster, Elies and Latorre, José I.},
   year={2020},
   month=feb, pages={226} }

@article{P_rez_Salinas_2021,
   title={One qubit as a universal approximant},
   volume={104},
   ISSN={2469-9934},
   url={http://dx.doi.org/10.1103/PhysRevA.104.012405},
   DOI={10.1103/physreva.104.012405},
   number={1},
   journal={Physical Review A},
   publisher={American Physical Society (APS)},
   author={Pérez-Salinas, Adrián and López-Núñez, David and García-Sáez, Artur and Forn-Díaz, P. and Latorre, José I.},
   year={2021},
   month=jul }

@article{Schuld_2022,
   title={Is Quantum Advantage the Right Goal for Quantum Machine Learning?},
   volume={3},
   ISSN={2691-3399},
   url={http://dx.doi.org/10.1103/PRXQuantum.3.030101},
   DOI={10.1103/prxquantum.3.030101},
   number={3},
   journal={PRX Quantum},
   publisher={American Physical Society (APS)},
   author={Schuld, Maria and Killoran, Nathan},
   year={2022},
   month=jul }

@article{Jerbi_2024,
   title={Shadows of quantum machine learning},
   volume={15},
   ISSN={2041-1723},
   url={http://dx.doi.org/10.1038/s41467-024-49877-8},
   DOI={10.1038/s41467-024-49877-8},
   number={1},
   journal={Nature Communications},
   publisher={Springer Science and Business Media LLC},
   author={Jerbi, Sofiene and Gyurik, Casper and Marshall, Simon C. and Molteni, Riccardo and Dunjko, Vedran},
   year={2024},
   month=jul }

@inproceedings{10.1145/2939672.2939778,
author = {Ribeiro, Marco Tulio and Singh, Sameer and Guestrin, Carlos},
title = {"Why Should I Trust You?": Explaining the Predictions of Any Classifier},
year = {2016},
isbn = {9781450342322},
publisher = {Association for Computing Machinery},
address = {New York, NY, USA},
url = {https://doi.org/10.1145/2939672.2939778},
doi = {10.1145/2939672.2939778},
booktitle = {Proceedings of the 22nd ACM SIGKDD International Conference on Knowledge Discovery and Data Mining},
pages = {1135–1144},
numpages = {10},
location = {San Francisco, California, USA},
series = {KDD '16}
}

@article{Pira_2024,
   title={On the interpretability of quantum neural networks},
   volume={6},
   ISSN={2524-4914},
   url={http://dx.doi.org/10.1007/s42484-024-00191-y},
   DOI={10.1007/s42484-024-00191-y},
   number={2},
   journal={Quantum Machine Intelligence},
   publisher={Springer Science and Business Media LLC},
   author={Pira, Lirandë and Ferrie, Chris},
   year={2024},
   month=aug }

@article{McClean_2016,
   title={The theory of variational hybrid quantum-classical algorithms},
   volume={18},
   ISSN={1367-2630},
   url={http://dx.doi.org/10.1088/1367-2630/18/2/023023},
   DOI={10.1088/1367-2630/18/2/023023},
   number={2},
   journal={New Journal of Physics},
   publisher={IOP Publishing},
   author={McClean, Jarrod R and Romero, Jonathan and Babbush, Ryan and Aspuru-Guzik, Alán},
   year={2016},
   month=feb, pages={023023} }

@article{PhysRevA.98.032309,
  title = {Quantum circuit learning},
  author = {Mitarai, K. and Negoro, M. and Kitagawa, M. and Fujii, K.},
  journal = {Phys. Rev. A},
  volume = {98},
  issue = {3},
  pages = {032309},
  numpages = {6},
  year = {2018},
  month = {Sep},
  publisher = {American Physical Society},
  doi = {10.1103/PhysRevA.98.032309},
  url = {https://link.aps.org/doi/10.1103/PhysRevA.98.032309}
}

@misc{cerezo2024doesprovableabsencebarren,
      title={Does provable absence of barren plateaus imply classical simulability? Or, why we need to rethink variational quantum computing}, 
      author={M. Cerezo and Martin Larocca and Diego García-Martín and N. L. Diaz and Paolo Braccia and Enrico Fontana and Manuel S. Rudolph and Pablo Bermejo and Aroosa Ijaz and Supanut Thanasilp and Eric R. Anschuetz and Zoë Holmes},
      year={2024},
      eprint={2312.09121},
      archivePrefix={arXiv},
      primaryClass={quant-ph},
      url={https://arxiv.org/abs/2312.09121}, 
}

@article{Larocca_2025,
   title={Barren plateaus in variational quantum computing},
   volume={7},
   ISSN={2522-5820},
   url={http://dx.doi.org/10.1038/s42254-025-00813-9},
   DOI={10.1038/s42254-025-00813-9},
   number={4},
   journal={Nature Reviews Physics},
   publisher={Springer Science and Business Media LLC},
   author={Larocca, Martín and Thanasilp, Supanut and Wang, Samson and Sharma, Kunal and Biamonte, Jacob and Coles, Patrick J. and Cincio, Lukasz and McClean, Jarrod R. and Holmes, Zoë and Cerezo, M.},
   year={2025},
   month=mar, pages={174–189} }

@misc{cotler2021revisitingdequantizationquantumadvantage,
      title={Revisiting dequantization and quantum advantage in learning tasks}, 
      author={Jordan Cotler and Hsin-Yuan Huang and Jarrod R. McClean},
      year={2021},
      eprint={2112.00811},
      archivePrefix={arXiv},
      primaryClass={quant-ph},
      url={https://arxiv.org/abs/2112.00811}, 
}

@article{Longo_2024,
   title={Explainable Artificial Intelligence (XAI) 2.0: A manifesto of open challenges and interdisciplinary research directions},
   volume={106},
   ISSN={1566-2535},
   url={http://dx.doi.org/10.1016/j.inffus.2024.102301},
   DOI={10.1016/j.inffus.2024.102301},
   journal={Information Fusion},
   publisher={Elsevier BV},
   author={Longo, Luca and Brcic, Mario and Cabitza, Federico and Choi, Jaesik and Confalonieri, Roberto and Ser, Javier Del and Guidotti, Riccardo and Hayashi, Yoichi and Herrera, Francisco and Holzinger, Andreas and Jiang, Richard and Khosravi, Hassan and Lecue, Freddy and Malgieri, Gianclaudio and Páez, Andrés and Samek, Wojciech and Schneider, Johannes and Speith, Timo and Stumpf, Simone},
   year={2024},
   month=jun, pages={102301} }

@book{molnar2025,
  title={Interpretable Machine Learning},
  subtitle={A Guide for Making Black Box Models Explainable},
  author={Christoph Molnar},
  year={2025},
  edition={3},
  isbn={978-3-911578-03-5},
  url={https://christophm.github.io/interpretable-ml-book}
}

@misc{dunjko2017machinelearningartificial,
      title={Machine learning \& artificial intelligence in the quantum domain}, 
      author={Vedran Dunjko and Hans J. Briegel},
      year={2017},
      eprint={1709.02779},
      archivePrefix={arXiv},
      primaryClass={quant-ph},
      url={https://arxiv.org/abs/1709.02779}, 
}

@article{Wang_2024,
   title={A comprehensive review of quantum machine learning: from NISQ to fault tolerance},
   volume={87},
   ISSN={1361-6633},
   url={http://dx.doi.org/10.1088/1361-6633/ad7f69},
   DOI={10.1088/1361-6633/ad7f69},
   number={11},
   journal={Reports on Progress in Physics},
   publisher={IOP Publishing},
   author={Wang, Yunfei and Liu, Junyu},
   year={2024},
   month=oct, pages={116402} }

@article{Bridgeman_2017,
   title={Hand-waving and interpretive dance: an introductory course on tensor networks},
   volume={50},
   ISSN={1751-8121},
   url={http://dx.doi.org/10.1088/1751-8121/aa6dc3},
   DOI={10.1088/1751-8121/aa6dc3},
   number={22},
   journal={Journal of Physics A: Mathematical and Theoretical},
   publisher={IOP Publishing},
   author={Bridgeman, Jacob C and Chubb, Christopher T},
   year={2017},
   month=may, pages={223001} }

@misc{perezgarcia2007matrixproductstaterepresentations,
      title={Matrix Product State Representations}, 
      author={D. Perez-Garcia and F. Verstraete and M. M. Wolf and J. I. Cirac},
      year={2007},
      eprint={quant-ph/0608197},
      archivePrefix={arXiv},
      primaryClass={quant-ph},
      url={https://arxiv.org/abs/quant-ph/0608197}, 
}

@misc{nair2025localsurrogatesquantummachine,
      title={Local surrogates for quantum machine learning}, 
      author={Sreeraj Rajindran Nair and Christopher Ferrie},
      year={2025},
      eprint={2506.09425},
      archivePrefix={arXiv},
      primaryClass={quant-ph},
      url={https://arxiv.org/abs/2506.09425}, 
}

@misc{chang2025primerquantummachinelearning,
      title={A Primer on Quantum Machine Learning}, 
      author={Su Yeon Chang and M. Cerezo},
      year={2025},
      eprint={2511.15969},
      archivePrefix={arXiv},
      primaryClass={quant-ph},
      url={https://arxiv.org/abs/2511.15969}, 
}

@ARTICLE{11014055,
  author={Lamichhane, Pradeep and Rawat, Danda B.},
  journal={IEEE Access}, 
  title={Quantum Machine Learning: Recent Advances, Challenges, and Perspectives}, 
  year={2025},
  volume={13},
  number={},
  pages={94057-94105},
  doi={10.1109/ACCESS.2025.3573244}}

@misc{acampora2025quantumcomputingartificialintelligence,
      title={Quantum computing and artificial intelligence: status and perspectives}, 
      author={Giovanni Acampora and Andris Ambainis and Natalia Ares and Leonardo Banchi and Pallavi Bhardwaj and Daniele Binosi and G. Andrew D. Briggs and Tommaso Calarco and Vedran Dunjko and Jens Eisert and Olivier Ezratty and Paul Erker and Federico Fedele and Elies Gil-Fuster and Martin Gärttner and Mats Granath and Markus Heyl and Iordanis Kerenidis and Matthias Klusch and Anton Frisk Kockum and Richard Kueng and Mario Krenn and Jörg Lässig and Antonio Macaluso and Sabrina Maniscalco and Florian Marquardt and Kristel Michielsen and Gorka Muñoz-Gil and Daniel Müssig and Hendrik Poulsen Nautrup and Sophie A. Neubauer and Evert van Nieuwenburg and Roman Orus and Jörg Schmiedmayer and Markus Schmitt and Philipp Slusallek and Filippo Vicentini and Christof Weitenberg and Frank K. Wilhelm},
      year={2025},
      eprint={2505.23860},
      archivePrefix={arXiv},
      primaryClass={quant-ph},
      url={https://arxiv.org/abs/2505.23860}, 
}

@misc{huang2025vastworldquantumadvantage,
      title={The vast world of quantum advantage}, 
      author={Hsin-Yuan Huang and Soonwon Choi and Jarrod R. McClean and John Preskill},
      year={2025},
      eprint={2508.05720},
      archivePrefix={arXiv},
      primaryClass={quant-ph},
      url={https://arxiv.org/abs/2508.05720}, 
}

@misc{eisert2025mindgapsfraughtroad,
      title={Mind the gaps: The fraught road to quantum advantage}, 
      author={Jens Eisert and John Preskill},
      year={2025},
      eprint={2510.19928},
      archivePrefix={arXiv},
      primaryClass={quant-ph},
      url={https://arxiv.org/abs/2510.19928}, 
}

@article{Huggins_2019,
   title={Towards quantum machine learning with tensor networks},
   volume={4},
   ISSN={2058-9565},
   url={http://dx.doi.org/10.1088/2058-9565/aaea94},
   DOI={10.1088/2058-9565/aaea94},
   number={2},
   journal={Quantum Science and Technology},
   publisher={IOP Publishing},
   author={Huggins, William and Patil, Piyush and Mitchell, Bradley and Whaley, K Birgitta and Stoudenmire, E Miles},
   year={2019},
   month=jan, pages={024001} }

@article{_unkovi__2022,
   title={Deep tensor networks with matrix product operators},
   volume={4},
   ISSN={2524-4914},
   url={http://dx.doi.org/10.1007/s42484-022-00081-1},
   DOI={10.1007/s42484-022-00081-1},
   number={2},
   journal={Quantum Machine Intelligence},
   publisher={Springer Science and Business Media LLC},
   author={Žunkovič, Bojan},
   year={2022},
   month=aug }

@article{Rieser_2023,
   title={Tensor networks for quantum machine learning},
   volume={479},
   ISSN={1471-2946},
   url={http://dx.doi.org/10.1098/rspa.2023.0218},
   DOI={10.1098/rspa.2023.0218},
   number={2275},
   journal={Proceedings of the Royal Society A: Mathematical, Physical and Engineering Sciences},
   publisher={The Royal Society},
   author={Rieser, Hans-Martin and Köster, Frank and Raulf, Arne Peter},
   year={2023},
   month=jul }

@article{Preskill_2018,
   title={Quantum Computing in the NISQ era and beyond},
   volume={2},
   ISSN={2521-327X},
   url={http://dx.doi.org/10.22331/q-2018-08-06-79},
   DOI={10.22331/q-2018-08-06-79},
   journal={Quantum},
   publisher={Verein zur Forderung des Open Access Publizierens in den Quantenwissenschaften},
   author={Preskill, John},
   year={2018},
   month=aug, pages={79} }

@article{RevModPhys.94.015004,
  title = {Noisy intermediate-scale quantum algorithms},
  author = {Bharti, Kishor and Cervera-Lierta, Alba and Kyaw, Thi Ha and Haug, Tobias and Alperin-Lea, Sumner and Anand, Abhinav and Degroote, Matthias and Heimonen, Hermanni and Kottmann, Jakob S. and Menke, Tim and Mok, Wai-Keong and Sim, Sukin and Kwek, Leong-Chuan and Aspuru-Guzik, Al\'an},
  journal = {Rev. Mod. Phys.},
  volume = {94},
  issue = {1},
  pages = {015004},
  numpages = {69},
  year = {2022},
  month = {Feb},
  publisher = {American Physical Society},
  doi = {10.1103/RevModPhys.94.015004},
  url = {https://link.aps.org/doi/10.1103/RevModPhys.94.015004}
}

@misc{gujju2024quantummachinelearningnearterm,
      title={Quantum Machine Learning on Near-Term Quantum Devices: Current State of Supervised and Unsupervised Techniques for Real-World Applications}, 
      author={Yaswitha Gujju and Atsushi Matsuo and Rudy Raymond},
      year={2024},
      eprint={2307.00908},
      archivePrefix={arXiv},
      primaryClass={quant-ph},
      url={https://arxiv.org/abs/2307.00908}, 
}

@article{Melnikov_2023,
   title={Quantum machine learning: from physics to software engineering},
   volume={8},
   ISSN={2374-6149},
   url={http://dx.doi.org/10.1080/23746149.2023.2165452},
   DOI={10.1080/23746149.2023.2165452},
   number={1},
   journal={Advances in Physics: X},
   publisher={Informa UK Limited},
   author={Melnikov, Alexey and Kordzanganeh, Mohammad and Alodjants, Alexander and Lee, Ray-Kuang},
   year={2023},
   month=feb }

@misc{preskill2025nisqmegaquopmachine,
      title={Beyond NISQ: The Megaquop Machine}, 
      author={John Preskill},
      year={2025},
      eprint={2502.17368},
      archivePrefix={arXiv},
      primaryClass={quant-ph},
      doi={https://doi.org/10.1145/3723153},
      url={https://arxiv.org/abs/2502.17368}, 
}

@misc{zimbors2025mythsquantumcomputationfault,
      title={Myths around quantum computation before full fault tolerance: What no-go theorems rule out and what they don't}, 
      author={Zoltán Zimborás and Bálint Koczor and Zoë Holmes and Elsi-Mari Borrelli and András Gilyén and Hsin-Yuan Huang and Zhenyu Cai and Antonio Acín and Leandro Aolita and Leonardo Banchi and Fernando G. S. L. Brandão and Daniel Cavalcanti and Toby Cubitt and Sergey N. Filippov and Guillermo García-Pérez and John Goold and Orsolya Kálmán and Elica Kyoseva and Matteo A. C. Rossi and Boris Sokolov and Ivano Tavernelli and Sabrina Maniscalco},
      year={2025},
      eprint={2501.05694},
      archivePrefix={arXiv},
      primaryClass={quant-ph},
      url={https://arxiv.org/abs/2501.05694}, 
}

@misc{mhiri2025unifyingaccountwarmstart,
      title={A unifying account of warm start guarantees for patches of quantum landscapes}, 
      author={Hela Mhiri and Ricard Puig and Sacha Lerch and Manuel S. Rudolph and Thiparat Chotibut and Supanut Thanasilp and Zoë Holmes},
      year={2025},
      eprint={2502.07889},
      archivePrefix={arXiv},
      primaryClass={quant-ph},
      url={https://arxiv.org/abs/2502.07889}, 
}

@misc{molteni2025quantummachinelearningadvantages,
      title={Quantum machine learning advantages beyond hardness of evaluation}, 
      author={Riccardo Molteni and Simon C. Marshall and Vedran Dunjko},
      year={2025},
      eprint={2504.15964},
      archivePrefix={arXiv},
      primaryClass={quant-ph},
      url={https://arxiv.org/abs/2504.15964}, 
}

@misc{salmon2023provableadvantagequantumpac,
      title={Provable Advantage in Quantum PAC Learning}, 
      author={Wilfred Salmon and Sergii Strelchuk and Tom Gur},
      year={2023},
      eprint={2309.10887},
      archivePrefix={arXiv},
      primaryClass={quant-ph},
      url={https://arxiv.org/abs/2309.10887}, 
}

@article{Tripathi:2025qki,
    author = "Tripathi, Sarvapriya and Upadhyay, Himanshu and Soni, Jayesh",
    title = {{Carbon efficient quantum AI: an empirical study of ans{\"a}tz design trade-offs in QNN and QLSTM models}},
    doi = "10.1038/s41598-025-28582-6",
    journal = "Sci. Rep.",
    volume = "15",
    number = "1",
    pages = "44936",
    year = "2025"
}

@misc{fellousasiani2022resourcecostlargescale,
      title={The resource cost of large scale quantum computing}, 
      author={Marco Fellous-Asiani},
      year={2022},
      eprint={2112.04022},
      archivePrefix={arXiv},
      primaryClass={quant-ph},
      url={https://arxiv.org/abs/2112.04022}, 
}

@misc{arunachalam2017surveyquantumlearningtheory,
      title={A Survey of Quantum Learning Theory}, 
      author={Srinivasan Arunachalam and Ronald de Wolf},
      year={2017},
      eprint={1701.06806},
      archivePrefix={arXiv},
      primaryClass={quant-ph},
      url={https://arxiv.org/abs/1701.06806}, 
}

@misc{banknote_authentication_267,
  author       = {Lohweg, Volker},
  title        = {{Banknote Authentication}},
  year         = {2012},
  howpublished = {UCI Machine Learning Repository},
  note         = {{DOI}: https://doi.org/10.24432/C55P57}
}

@misc{hernicht2025enhancingscalabilityclassicalsurrogates,
      title={Enhancing the Scalability of Classical Surrogates for Real-World Quantum Machine Learning Applications}, 
      author={Philip Anton Hernicht and Alona Sakhnenko and Corey O'Meara and Giorgio Cortiana and Jeanette Miriam Lorenz},
      year={2025},
      eprint={2508.06131},
      archivePrefix={arXiv},
      primaryClass={quant-ph},
      url={https://arxiv.org/abs/2508.06131}, 
}

@misc{hangleiter2026quantumadvantageachieved,
      title={Has quantum advantage been achieved?}, 
      author={Dominik Hangleiter},
      year={2026},
      eprint={2603.09901},
      archivePrefix={arXiv},
      primaryClass={quant-ph},
      url={https://arxiv.org/abs/2603.09901}, 
}

@misc{eisert2013entanglementtensornetworkstates,
      title={Entanglement and tensor network states}, 
      author={J. Eisert},
      year={2013},
      eprint={1308.3318},
      archivePrefix={arXiv},
      primaryClass={quant-ph},
      url={https://arxiv.org/abs/1308.3318}, 
}

@misc{rodrguezaldavero2025chebyshevapproximationcompositionfunctions,
      title={Chebyshev approximation and composition of functions in matrix product states for quantum-inspired numerical analysis}, 
      author={Juan José Rodríguez-Aldavero and Paula García-Molina and Luca Tagliacozzo and Juan José García-Ripoll},
      year={2025},
      eprint={2407.09609},
      archivePrefix={arXiv},
      primaryClass={quant-ph},
      url={https://arxiv.org/abs/2407.09609}, 
}

@article{Oseledets2011,
  title        = {Tensor-Train Decomposition},
  author       = {Oseledets, Ivan V.},
  journal      = {SIAM Journal on Scientific Computing},
  volume       = {33},
  number       = {5},
  pages        = {2295--2317},
  year         = {2011},
  doi          = {10.1137/090752286},
  publisher    = {Society for Industrial and Applied Mathematics (SIAM)}
}

@article{Caro2022,
  title        = {Generalization in Quantum Machine Learning from Few Training Data},
  author       = {Caro, Manuel and Huang, Hsin-Yuan and Cerezo, M. and Sharma, Kunal and Sornborger, Andrew and Cincio, Lukasz and Coles, Patrick J.},
  journal      = {Nature Communications},
  volume       = {13},
  pages        = {4919},
  year         = {2022},
  doi          = {10.1038/s41467-022-32550-3},
  publisher    = {Springer Nature}
}

@inproceedings{KhavariRabusseau2021,
 author = {Khavari, Behnoush and Rabusseau, Guillaume},
 booktitle = {Advances in Neural Information Processing Systems},
 editor = {M. Ranzato and A. Beygelzimer and Y. Dauphin and P.S. Liang and J. Wortman Vaughan},
 pages = {10931--10943},
 publisher = {Curran Associates, Inc.},
 title = {Lower and Upper Bounds on the Pseudo-Dimension of Tensor Network Models},
 url = {https://proceedings.neurips.cc/paper_files/paper/2021/file/5a9d8bf5b7a4b35f3110dde8673bdda2-Paper.pdf},
 volume = {34},
 year = {2021}
}

@book{Mohri2018,
  title        = {Foundations of Machine Learning},
  author       = {Mohri, Mehryar and Rostamizadeh, Afshin and Talwalkar, Ameet},
  publisher    = {MIT Press},
  edition      = {2},
  year         = {2018},
  isbn         = {9780262039406}
}

@book{ShalevShwartzBenDavid2014,
  title        = {Understanding Machine Learning: From Theory to Algorithms},
  author       = {Shalev-Shwartz, Shai and Ben-David, Shai},
  publisher    = {Cambridge University Press},
  year         = {2014},
  doi          = {10.1017/CBO9781107298019},
  isbn         = {9781107057135}
}

@article{sauer1972density,
  title        = {On the Density of Families of Sets},
  author       = {Sauer, N.},
  journal      = {Journal of Combinatorial Theory, Series A},
  volume       = {13},
  number       = {1},
  pages        = {145--147},
  year         = {1972},
  doi          = {10.1016/0097-3165(72)90019-2},
  publisher    = {Elsevier}
}

@article{HiriartUrrutyLe2013rank,
  title        = {A Variational Approach of the Rank Function},
  author       = {Hiriart-Urruty, Jean-Baptiste and Le, Hoang Tuy},
  journal      = {TOP},
  volume       = {21},
  number       = {2},
  pages        = {207--240},
  year         = {2013},
  doi          = {10.1007/s11750-012-0234-4},
  publisher    = {Springer}
}

@article{BlumerEtAl1989,
  title        = {Learnability and the {Vapnik-Chervonenkis} Dimension},
  author       = {Blumer, Anselm and Ehrenfeucht, Andrzej and Haussler, David and Warmuth, Manfred K.},
  journal      = {Journal of the ACM},
  volume       = {36},
  number       = {4},
  pages        = {929--965},
  year         = {1989},
  doi          = {10.1145/76359.76371},
  publisher    = {Association for Computing Machinery (ACM)}
}

@ARTICLE{788640,
  author={Vapnik, V.N.},
  journal={IEEE Transactions on Neural Networks}, 
  title={An overview of statistical learning theory}, 
  year={1999},
  volume={10},
  number={5},
  pages={988-999},
  doi={10.1109/72.788640}}

@book{Vapnik1999-VAPTNO,
	author = {Vladimir Vapnik},
	editor = {},
	publisher = {Springer: New York},
	title = {The Nature of Statistical Learning Theory},
	year = {1999}
}

@article{JMLR:v11:shalev-shwartz10a,
  author  = {Shai Shalev-Shwartz and Ohad Shamir and Nathan Srebro and Karthik Sridharan},
  title   = {Learnability, Stability and Uniform Convergence},
  journal = {Journal of Machine Learning Research},
  year    = {2010},
  volume  = {11},
  number  = {90},
  pages   = {2635--2670},
  url     = {http://jmlr.org/papers/v11/shalev-shwartz10a.html}
}

@misc{watanabe2026tensornetworksurrogatemodels,
      title={Tensor network surrogate models for variational quantum computation}, 
      author={Ryo Watanabe and Dries Sels and Joseph Tindall},
      year={2026},
      eprint={2604.20180},
      archivePrefix={arXiv},
      primaryClass={quant-ph},
      url={https://arxiv.org/abs/2604.20180}, 
}


\appendix
\renewcommand{\thesection}{\Alph{section}}

\section{Appendix: VC Dimension Union Bound}

We provide a self-contained proof of the union bound for VC dimension used in Lemma \ref{lem:loss_pdim_sq}.

\begin{lemma}[VC Dimension of Union]
\label{lem:vc-union}
Let $\mathcal{C}_1, \mathcal{C}_2$ be set families over a domain $\mathcal{X}$ with $\mathrm{VCdim}(\mathcal{C}_1) = d_1$ and $\mathrm{VCdim}(\mathcal{C}_2) = d_2$. Then
\[
\mathrm{VCdim}(\mathcal{C}_1 \cup \mathcal{C}_2) \leq d_1 + d_2 + 1.
\]
\end{lemma}

\begin{proof}
For a finite set $S\subseteq\mathcal X$, let $\mathcal C|_S:=\{C\cap S: C\in\mathcal C\}$ and
$\Pi_{\mathcal C}(n):=\max_{|S|=n} |\mathcal C|_S|$.

We use the standard growth function argument. For any $S$ with $|S|=n$, we have $(\mathcal C_1\cup\mathcal C_2)|_S=\mathcal C_1|_S\cup \mathcal C_2|_S$, hence
$|(\mathcal C_1\cup\mathcal C_2)|_S|\le |\mathcal C_1|_S|+|\mathcal C_2|_S|$.
Taking the maximum over $S$ gives
\[
\Pi_{\mathcal{C}_1 \cup \mathcal{C}_2}(n) \leq \Pi_{\mathcal{C}_1}(n) + \Pi_{\mathcal{C}_2}(n).
\]

By Sauer's lemma \cite{sauer1972density}, if $\mathrm{VCdim}(\mathcal{C}_i) = d_i$, then for all $n$,
\[
\Pi_{\mathcal{C}_i}(n) \leq \sum_{j=0}^{d_i} \binom{n}{j}.
\]

Assume $d_1 \leq d_2$ and set $n = d_1 + d_2 + 2$. Then
\[
\sum_{j=0}^{d_2} \binom{n}{j}
= 2^n - \sum_{j=d_2+1}^{n} \binom{n}{j}
= 2^n - \sum_{j=0}^{n-d_2-1} \binom{n}{j}
= 2^n - \sum_{j=0}^{d_1+1} \binom{n}{j},
\]
where we used $\binom{n}{j}=\binom{n}{n-j}$ and $n-d_2-1=d_1+1$.

Therefore
\begin{align*}
\Pi_{\mathcal{C}_1 \cup \mathcal{C}_2}(n) 
&\leq \sum_{j=0}^{d_1} \binom{n}{j} + \sum_{j=0}^{d_2} \binom{n}{j} \\
&\le \sum_{j=0}^{d_1} \binom{n}{j} + 2^n - \sum_{j=0}^{d_1+1} \binom{n}{j} \\
&= 2^n - \binom{n}{d_1+1}
< 2^n.
\end{align*}
Hence $\Pi_{\mathcal{C}_1 \cup \mathcal{C}_2}(n) < 2^n$, so $\mathcal{C}_1 \cup \mathcal{C}_2$ cannot shatter $n$ points.
Therefore:

\[
\mathrm{VCdim}(\mathcal{C}_1 \cup \mathcal{C}_2) \leq n - 1 = d_1 + d_2 + 1.
\]
\end{proof}

(See \cite[Ch. 6, Exercises, item 2]{ShalevShwartzBenDavid2014}, which asks to prove the
special case $\mathrm{VCdim}(H_1\cup H_2)\le 2d+1$)

\end{document}